\pdfoutput=1
\documentclass[runningheads]{llncs}
\usepackage{graphicx}
\usepackage{booktabs} 
\usepackage{comment}
\usepackage{amssymb}
\usepackage{amsmath}
\usepackage{amsthm, bm}
\usepackage{flushend}
\usepackage{url}
\newcommand{\ie}{{\em i.e.}}
\newcommand{\eg}{{\em e.g.}}
\newcommand{\et}{{\em et al.}}

\DeclareMathOperator*{\argmax}{arg\,max}

\begin{document}
\title{Optimal Search Segmentation Mechanisms for \\Online Platform Markets}
%
%
\author{Zhenzhe Zheng\inst{1} \and
R. Srikant\inst{2} }
\authorrunning{Z. Zheng and R. Srikant}
%
\institute{Coordinated Science Lab, University of Illinois at Urbana-Champaign
\email{zhenzhe@illinois.edu}\\
 \and
Coordinated Science Lab, Department of Electrical and Computer Engineering, University of Illinois at Urbana-Champaign\\
\email{rsrikant@illinois.edu}}
\maketitle              
\begin{abstract}
Online platforms, such as Airbnb, hotels.com, Amazon, Uber and Lyft, can control and optimize many aspects of product search to improve the efficiency of marketplaces. Here we focus on a common model, called the discriminatory control model, where the platform chooses to display a subset of sellers who sell products at prices determined by the market and a buyer is interested in buying a single product from one of the sellers. Under the commonly-used model for single product selection by a buyer, called the multinomial logit model, and the Bertrand game model for competition among sellers, we show the following result: to maximize social welfare, the optimal strategy for the platform is to display all products; however, to maximize revenue, the optimal strategy is to only display a subset of the products whose qualities are above a certain threshold.
We  extend our results to Cournot competition model, and show that the optimal search segmentation mechanisms for both social welfare maximization and revenue maximization also have simple threshold structures. The threshold in each case  depends on the quality of all products, the platform's objective and seller's competition model, and can be computed  in linear time in the number of products.
\end{abstract}
\section{Introduction}
In recent years, we have witnessed the rise of many successful online platform markets, which have reshaped the economic landscape of modern world. The online platforms facilitate the exchange of goods and services between buyers and sellers. For example, buyers can purchase goods from sellers on Amazon, eBay and Etsy, arrange accommodation from hosts on Airbnb and Expedia, order transportation services from drivers on Uber and Lyft, and find qualified workers on online labor markets, such as Upwork and Taskrabbit. The total market value of online platforms has exceeded 4.3 trillion dollars worldwide, and is growing quickly~\cite{evans2016rise}.

One salient feature of these online platforms are that the market operators have fine-grained information about the underlying characteristics of transactions, and can leverage this knowledge to design effective and efficient market structure.
Compared with traditional markets, the modern online marketplaces have greater controls over price determination, search and discovery, information revelation, recommendation, etc. For example, Uber and Lyft adopt the \emph{full control model}, in which the ride-sharing platforms use online matching algorithms to determine matches between drivers and riders as well as the fee for the route~\cite{cannon2014uber,Chen:2015:PBH:2815675.2815681}. Amazon and Airbnb use the \emph{discriminatory control model}, where the platforms only control the list of products to display for  each buyer's search, and the potential matches and transaction prices are determined by the preference of buyers and the competition among sellers~\cite{Chen:2016:EAA,Grbovic:2018}. The platform can also use other types of control, such as commissions/subscriptions fees~\cite{Birge:2018:OCS:3219166.3219216}, to influence the outcomes of markets. The rich control options for online platforms have led to an increasing discussion about the design of online marketplaces with different optimization objectives~\cite{Arnosti:2014:MCD:2600057.2602893,Banerjee:2017:STM:3038912.3052578,kanoria2017facilitating,NBERw24282}.

In this paper, we investigate  social welfare and revenue optimization under the discriminatory control model for
online marketplaces.
In the discriminatory control model,  the platform has only control over \emph{search segmentation mechanisms - which products to display for each buyer's search}, and the transaction prices are endogenously determined by the competition among sellers.
Unlike traditional firms, most online platforms do not manufacture goods or provide services, and thus they also do not dictate the specific transaction prices. Instead, buyers and sellers jointly determine the prices at which the goods or services will be traded. For example, sellers set prices for their goods on Amazon, hosts decide on the price for their properties on Airbnb, and freelancers negotiate employers with hourly fee on Upwork. These prices depend on the demand and supply for comparable goods and services in the market, and
choosing different products to display for buyers impacts the transaction prices and then the social welfare and revenue.
Motivated by this, we study the role of search segmentation mechanisms in social welfare and revenue optimization in  the discriminatory control model with endogenous prices.

To calculate the  social welfare and revenue, we first need to specify  demand and supply in online marketplaces.
Much of prior work simply represent the demand/supply curves with non-increasing/non-decreasing distributions~\cite{Banerjee:2017:STM:3038912.3052578,Birge:2018:OCS:3219166.3219216}. Instead, we consider a demand and supply function derived from a basic market setting in which each seller has one unit of product to offer, and each buyer demands at most one unit of product chosen from the products displayed to her\footnote{Throughout the article, we use product to refer good/service, and use the terms of product and seller interchangeably.}. Given the quality and prices of products, the demand for each product is equivalent to the proportion of potential buyers that purchase such a product. Thus, the demand function is closely related to the purchase behaviors of buyers who face multiple substitutable products. We adopt the standard multinomial logit (MNL) model~\cite{McFa73} to describe buyers' choice behaviors, and then derive the demand as a softmax function.
With such a specific demand function, we can model the competition among sellers via a Bertrand price competition game, which is a useful model for investigating oligopolistic competition in real markets~\cite{vives2001oligopoly}.
For instance, the Bertrand game can model the situation where the hosts on Airbnb compete for potential guests by  setting prices for their properties.
The basic questions for the Bertrand competition game are  existence, uniqueness, closed-form expression and learning algorithm  of the equilibrium.
The results in~\cite{doi:10.1287/mnsc.1120.1664,doi:10.1287/msom.1060.0115}  have shown that there exists a unique (pure) Nash equilibrium in the Bertrand game with a MNL model.
Furthermore, 
the Nash equilibrium coincides with the solution of a system of first-order-condition equations.
We can then characterize the Nash equilibrium
in a ``closed'' form, and express the equilibrium social welfare/revenue by employing a variant of Lambert W function~\cite{Corless1996}. We also derive myopic learning strategies, \ie, best response dynamics, for sellers to reach the Nash equilibrium in practice.

The online platform can further optimize the equilibrium social welfare/revenue by employing search segmentation mechanisms.
Different sets of sellers involved in the Bertrand competition game lead to different equilibrium solutions. The goal of the search segmentation mechanisms is to efficiently choose a set of products to display for buyers (or in other words, choose a set of sellers to compete in the Bertrand game) that maximizes the equilibrium social welfare or revenue. This display control optimization problem is combinatorial in nature and the number of possible product sets can be very large,  particularly when there are many potential products to offer. One of our main contributions is to identify the efficient and optimal search segmentation mechanism, which turns out to have a simple structure. We show that the online platform \emph{will display all products to maximize equilibrium social welfare, but just display the top $k^*$ highest quality products to maximize equilibrium revenue.}  We also refer the optimal mechanism for revenue maximization as \emph{quality-order mechanism.} 
The optimal threshold $k^*$ depends on the quality of all products, and can be calculated in linear time in terms of the number of products. 
The optimality of such simple search segmentation mechanisms has crucial theoretical and practical implications. On the theoretical side, this result allows the platform to find the optimal set of displayed products in linear time, significantly reducing the computational complexity of searching for the optimal solution. On the practical side, optimality of quality-order mechanism is quite appealing as it guarantees that a lower quality product will not be chosen for display over a higher quality product.
Moreover, in order to increase the opportunity of being selected, sellers would improve the quality of their products as product quality is the selection criteria of the optimal mechanism, which will benefit all the market participants in the long term.

The  optimality of the quality-order mechanism for revenue maximization  is established by making a novel connection between the quasi-convexity of equilibrium revenue functions and the optimal control decision on selecting displayed products.
We show that in the Bertrand game with a given subset of sellers,
the equilibrium revenue can be expressed as a quasi-convex function with respective to an independent variable, which is a one-to-one transformation of the quality of a candidate product.
The property of quasi-convexity guarantees  that the maximum revenue can be obtained at one of the two endpoints, which corresponds to the options of displaying the current set of products or involving a new product with the highest quality among the remaining products.
With this critical observation, if the platform decides to add a new product, it will always select the available product with the highest quality. Thus, we can efficiently construct the optimal set of displayed products from any product set. Specifically, if the current product set does not contain all the top $k^*$ products, we can further improve the equilibrium revenue by repeatedly replacing one currently selected product with an unselected product with a higher quality.


Our work in this paper is related to work on the design of markets for networked platforms~\cite{Banerjee:2017:STM:3038912.3052578,Bimpikis:2014:CCN:2600057.2602882,8264340,8057125}. We present a detailed discussion of related work towards the end of the paper. Here, we briefly discuss the similarities and differences between our work and prior work on networked market platforms. In networked markets, there are buyers and sellers connected by a bipartite graph, where each link indicates that a specific buyer is allowed to buy from a specific seller. The goal is to remove links from the complete bipartite graph to maximize either social welfare or revenue. However, much of the prior work focuses on a linear price-demand curve which does not explicitly model situations where each buyer is interested in buying only one product (such as one copy of a book) and each buyer takes into account the quality of each product (available typically in the form of reviews) while making a buying decision. For such situations, economists use the MNL model, which we have adopted in this paper. On the other hand, compared to prior work on networked markets, we only consider a much simpler bipartite graph where there is only one representative buyer. Such a model is appropriate when there are no capacity constraints for products at a seller, for example, each seller may have many copies of a book and there is no danger of immediately selling out a particular book title. The model is also appropriate for hotels.com-type settings in situations when most hotels have multiple available rooms. In situations where multiple buyers are performing searches simultaneously and hotels are about to sell out of rooms, capacity constraints do matter. Such capacity-constrained situations have not been studied either in this paper or in prior work, and is a topic for future research.

We now summarize the main contributions of this paper.

$\bullet$ We introduce a stylized model to capture the main features of online platform markets.
We explicitly model the market, where each buyer is interested in purchasing  one product, and takes  into account the quality of products when making choice. Specifically, the demand function for products is derived from the multinomial logit (MNL) choice model, and the supply response of sellers
is described by the outcome of Bertrand competition game. We show that the Bertrand game exists a unique (pure) Nash equilibrium, and the best response dynamics converge to the equilibrium. We also explicitly express the social welfare and revenue under the equilibrium.


$\bullet$ We design efficient search segmentation mechanisms to optimize equilibrium social welfare and revenue under the Bertrand  model of competition.  
We first prove that it is optimal to display all products to maximize social welfare. 
For revenue maximization, we then show that the optimal mechanism, referred to as quality-order mechanism,
only needs to display the top $k^*$ highest quality products, where the optimal number of products $k^*$ can be found in linear time. 

$\bullet$ We prove the result for social welfare maximization by showing the equilibrium social welfare function is  decreasing with respective to an independent variable, which also decreases for involving a new product. 
We establish the optimality of the quality-order mechanism for revenue maximization  by making a novel connection between the quasi-convexity of equilibrium revenue function and the optimal decision on selecting displayed products.

$\bullet$ We  extend our results to another classical oligopolistic competition model: Cournot competition game.  We show that the optimal search segmentation mechanisms for both social welfare and revenue maximization in this model also falls into the simple quality-order mechanisms, in which the optimal threshold $k^*$ depends on the product quality and the platform's specific objective. 







\section{Preliminaries}\label{sec:preli}
 We consider a two-sided market with $n$ sellers $\mathbb{S}= \{1, 2, \cdots, n\}$ and one \emph{representative} buyer, representing a set of homogeneous buyers. Each seller $i\in \mathbb{S}$ offers a product with quality $\theta_i$ and price $p_i$. We denote the quality and price vectors by $\boldsymbol{\theta}=(\theta_1, \theta_2, \cdots, \theta_n)$ and  $\boldsymbol{p} = (p_1, p_2, \cdots, p_n)$, respectively.
 The quality vector $\boldsymbol{\theta}$ is fixed, while the price vector $\boldsymbol{p}$ is determined by the competition among sellers.
Without loss of generality, we assume the products' quality and prices are non-negative, \ie, $\theta_i\geq 0$ and $p_i \geq 0$, and the sellers  are sorted according to the product quality in a non-decreasing order, \ie, $\theta_1 \geq \theta_2 \geq  \cdots \geq \theta_n$. Given the quality $\boldsymbol{\theta}$ and prices $\boldsymbol{p}$ of all products, the buyer purchases one of the $n$ products, or adopts an outside option, \ie, buys nothing from this market. We normalize the problem parameters so that outside option's quality $\theta_0$ and price $p_0$ are zero, \ie, $\theta_0 = p_0 =0$.

 In the  random utility model~\cite{10.2307/184004}, the buyer derives utility $u_i$ from purchasing the product $i\in \mathbb{S}$ or selecting the outside option $i=0$ as follows
\begin{equation*}
u_i \triangleq \theta_i + \xi_i - p_i,
\end{equation*}
where $\xi_i$ is a random variable representing buyer's (private) preference about  the $i$th alternative.
Given the $n+1$ choices ($n$  products and the outside option), the buyer selects the alternative with the maximum {utility}. Under the standard assumption that the random variables $\{ \xi_i \}$ are independent and identically distributed (i.i.d.) with Gumbel distribution~\cite{anderson1992discrete,doi:10.1287/mksc.2.3.203}, it can be shown~\cite{anderson1992discrete,McFa73} that the buyer selects the alternative $i\in \{0\}\cup \mathbb{S}$ with probability
\begin{equation}\label{demand}
q_i(\boldsymbol{p}) \triangleq  Pr(u_i = \max_{j\in  \{0\}\cup \mathbb{S}} u_j) =\frac{a_i}{1+ \sum_{j\in \mathbb{S}} a_j},
\end{equation}
where $a_i= \exp(\theta_i-p_i)$ for all $i \in \mathbb{S}$. We refer to $q_i$ as \emph{demand} or \emph{market share} of the alternative $i \in \{0\} \cup \mathbb{S}$. We can also interpret $q_i$ as the expected sales of quantity of product $i$ normalized by the total number of potential buyers.
This choice model is known as multinomial logit (MNL) model in economic literature~\cite{anderson1992discrete,doi:10.1287/mksc.2.3.203,McFa73}. We use $\boldsymbol{q} = (q_0, q_1, \cdots, q_n)$ to denote the demands of  products.

Under the above model, we can also obtain an explicit form for the utility $\bar{u}$ of the representative buyer
\begin{equation*}
\bar{u} \triangleq \mathbb{E}[\max_{i\in \{0\} \cup \mathbb{S}} u_i] = \log(1+\sum_{i \in \mathbb{S}} a_i).
\end{equation*}
From the demand $q_i(\bm{p})$ in (\ref{demand}), we can express seller $i$'s expected revenue $r_i(\bm{p})$ in terms of prices
\begin{equation}\label{eqn:seller_revenue}
r_i(\boldsymbol{p}) \triangleq p_i \times q_i(\bm{p}) = p_i \times \frac{a_i}{1+ \sum_{j\in \mathbb{S}} a_j}.
\end{equation}
The social welfare of the two-sided market is measured by the sum of buyer's utility and the total revenue of sellers, \ie,
\begin{equation}\label{eqn:sw_price}
sw(\boldsymbol{p}) \triangleq \bar{u}+ \sum_{i\in \mathbb{S}} r_i(\boldsymbol{p}) = \log(1+\sum_{j\in \mathbb{S}} a_j) + \sum_{i \in \mathbb{S}} p_i \times \frac{a_i}{1+ \sum_{j \in \mathbb{S}} a_j}.
\end{equation}
The revenue of the market is the total revenue of all sellers, \ie,
\begin{equation}\label{eqn:revenue_price}
{re}({\bm p}) \triangleq \sum_{i \in \mathbb{S}} r_i({\bm p}) = \sum_{i \in \mathbb{S}} p_i \times \frac{a_i}{1+ \sum_{j \in \mathbb{S}} a_j}.
\end{equation}
We now note  the relation between price and demand in the MNL model, which would be quite useful for  optimization and analysis later. Using the price-demand model  in (\ref{demand}), we can express the price $p_i$ in terms of  demands $\boldsymbol{q}$:
\begin{equation}\label{eqn:price}
p_i(\boldsymbol{q}) = \theta_i + \log(1- \sum_{j \in {S}} q_j) - \log(q_i).
\end{equation}
The social welfare and revenue optimization would become convenient if we work with the demands $\textbf{q}$ rather than the prices $\textbf{p}$. For example, the social welfare  and revenue functions  are not concave in $\textbf{p}$~\cite{doi:10.1287/mnsc.42.7.992}, but  become jointly concave if we express the functions in terms of $\textbf{q}$~\cite{doi:10.1287/msom.1080.0221,song2007demand}. In Appendix~\ref{sec:monopolistic_market}, we leverage this property to derive the optimal prices for social welfare and revenue maximization in the full control model, where the platform can control both price and displayed products. 
\section{Bertrand Competition Game}\label{sec:oligopolistic}
In  discriminatory control model, the online platform can only control the list of products to display for buyers, and the transaction prices are endogenously determined by the oligopolistic competition among sellers.
In a Bertrand competition game, the seller of each product sets a price. Based on the prices of the products and the set of available products, the market produces a certain demand for each product. In our MNL model, the demand is simply the probability with which a product will be purchased by the buyer. This is the typical situation in a Airbnb-like model, where the owner of each rental unit sets a price, the platform controls the manner in which the rental units are displayed, and the renter selects a unit to rent.

In this section, we investigate the existence and uniqueness of equilibrium in the Bertrand  competition game, explicitly express the equilibrium social welfare/revenue, and derive the  best response dynamics to reach the Nash equilibrium.
We assume that only a subset $S \subseteq \mathbb{S}$ of sellers are involved in the game. In other words, we assume that the platform has chosen to display the products of a subset $S$ of the sellers. In the next section, we will show how the choice of  $S$ can be optimized by the platform to maximize either social welfare or revenue.



In the Bertrand competition game, seller $i\in {S}$ selects price $p_i$ to maximize her revenue $r_i(\bm{p})= p_i \times q_i(\bm{p})$, where the demand $q_i(\bm{p})$ is determined by the prices $\bm{p}$ of all products in (\ref{demand}). We can formally represent the Bertrand game as a triplet $G^b=\left({S}, (\mathcal{P}_i)_{i\in {S}}, (r_i)_{i\in {S}}\right)$, where ${S}$ is a set of players, $\mathcal{P}_i$ is the strategy space of player $i\in {S}$ (\ie, $\mathcal{P}_i \triangleq \{p_i | p_i \geq 0\}$), and $r_i(\bm{p})$ is the payoff of player $i\in {S}$. We represent the set of strategy profiles by $\mathcal{P}=\mathcal{P}_1\times \mathcal{P}_2 \times \cdots \times \mathcal{P}_n$. We also denote the strategy profile $\bm{p}\in \mathcal{P}$ as $\bm{p}=(p_i, \bm{p}_{-i})$, where $\bm{p}_{-i}$ is the strategies (or prices) of all the players except $i$.
For such Bertrand game, we have the following result from~\cite{doi:10.1287/msom.1060.0115}.
\begin{theorem}
There exists a unique (pure) Nash equilibrium in the Bertrand game $G^b=\left({S}, (\mathcal{P}_i)_{i\in {S}}, (r_i)_{i\in {S}}\right)$. A vector of prices $\bm{\bar{p}}=(\bar{p}_1,\bar{p}_2, \cdots, \bar{p}_n) \in \mathcal{P}$ satisfies $\partial r_i(\bm{\bar{p}})/ \partial p_i =0$ for all $i\in {S}$ if and only if $\bm{\bar{p}}$ is a Nash equilibrium in $\mathcal{P}$.
\end{theorem}

We next calculate a closed-form expression for the Nash equilibrium prices $\bm{\bar{p}}$. For each seller $i\in {S}$, by the first-order condition $\partial r_i(\bm{\bar{p}})/ \partial p_i =0$, we have the following relation for $\bar{p}_i$:
\begin{equation}\label{eqn:derivation}
\bar{p}_i = \frac{1+ \sum_{j\in S} \bar{a}_j}{1+ \sum_{j\in S} \bar{a}_j -\bar{a}_i} =\frac{1}{1-\bar{q}_i},
\end{equation}
where $\bar{a}_i  \triangleq \exp(\theta_i-\bar{p}_i)$ and $\bar{q}_i$ is the demand of product $i$ at the equilibrium, \ie, $\bar{q}_i \triangleq {\bar{a}_i}/(1+ \sum_{j \in {S}} \bar{a}_j)$.
From the price function in (\ref{eqn:price}) and with some calculations applied to (\ref{eqn:derivation}), we  have the following equations
\begin{equation}\label{eqn:price_competition}
 \bar{q}_0 \times \exp(\theta_i-1) = \bar{q}_i \times \exp(\frac{\bar{q}_i}{1-\bar{q}_i}), \quad \forall i\in {S},
\end{equation}
where $\bar{q}_0 \triangleq 1-\sum_{j \in {S}} \bar{q}_j$ is the probability of the buyer that purchases nothing. We introduce a function $V(x): (0, +\infty) \rightarrow (0,1)$, such that for any $x\in (0, \infty)$, $V(x)$ is the solution $v\in (0,1)$ satisfying
\begin{equation}\label{eqn:v_func}
v\times \exp(\frac{v}{1-v}) = x.
\end{equation}
We can verify that $V(x)$ is a strictly increasing and concave function over $[0,+\infty)$.
This function is similar to the Lambert function $W(x)$~\cite{Corless1996}, which is the solution $w$ satisfying $w\times exp(w) = x$.
With the function $V(x)$ and (\ref{eqn:price_competition}), we can obtain a closed-form expression for the demand $\bar{q}_i = V(\bar{q}_0\times \exp(\theta_i-1)).$
Combing with the definition of $\bar{q}_0$, we can determine $\bar{q}_0$ by solving the following single-variable equation
\begin{equation}\label{eqn:non_purchase}
\sum_{i \in {S}} V(\bar{q}_0 \times \exp(\theta_i-1)) = 1- \bar{q}_0.
\end{equation}
This equation has a unique solution because $V(x)$ is a strictly increasing function. We also refer this equation as the equilibrium constraint.
The next theorem presents a closed-form expression for the Nash equilibrium solution in the Bertrand competition  game.
\begin{theorem}
In the Bertrand game $G^b=\left({S}, (\mathcal{P}_i)_{i\in {S}}, (r_i)_{i\in {S}}\right)$, the Nash equilibrium price $\bar{p}_i$ and the demand $\bar{q}_i$ for each product $i\in {S}$ are given by
$$
\bar{p}_i  = \frac{1}{1-V(\bar{q}_0\times \exp(\theta_i-1))} \quad \text{and} \quad \bar{q}_i = V(\bar{q}_0\times \exp(\theta_i-1)),
$$
where $\bar{q}_0$ is the unique solution to (\ref{eqn:non_purchase}).
\end{theorem}

Substituting the equilibrium solutions into (\ref{eqn:sw_price}), we obtain the equilibrium social welfare in the Bertrand game with the sellers $S\subseteq \mathbb{S}$
\begin{equation}\label{nsw_bertrand}
\overline{sw}(S)= -\log\left(\bar{q}_0 \right) + \sum_{i \in {S}} \frac{\bar{q}_i }{1-\bar{q}_i}.
\end{equation}
By (\ref{eqn:revenue_price}), we can similarly get the equilibrium revenue in the Bertrand game with the set of sellers $S\subseteq \mathbb{S}$
\begin{equation}
\overline{re}(S) =  \sum_{i \in {S}} \frac{\bar{q}_i }{1-\bar{q}_i}.
\end{equation}

Instead of directly deriving the equilibrium strategies in one single step, in practice, the sellers may employ some simple,  natural and myopic learning algorithms, such as best response~\cite{fudenberg1998theory}, fictitious play~\cite{1310468} or no-regret learning algorithm~\cite{FREUND199979}, to interact with each other and eventually reach the equilibrium.
One straightforward procedure for sellers in online platform markets to reach the Nash equilibrium is best response dynamics.
Specifically, suppose that the current vector of price $\textbf{p}$ is not a Nash equilibrium,  and a seller $i \in S$ deviates by setting a new $p^*_i$, which is the optimal price with respective to the other prices $\textbf{p}_{-i}$, \ie,
$$
p^*_i = B(\textbf{p}_{-i}) \triangleq \argmax_{p\in [0,+\infty)}  r_i(p, \textbf{p}_{-i}).
$$
We can verify that the revenue function $r_i(p, \textbf{p}_{-i})$ is strictly quasi-concave in $p$, and thus it is not easy to explicitly solve the above optimization problem.
One key observation is that the revenue function is strictly concave in the domain of the demand,
which enables us to obtain closed-form expressions for the best response strategies, as shown in the following lemma.
\begin{lemma}\label{lem:best_response}
The best response price $p^*_i$ with respective to a fixed price vector $\textbf{p}_{-i}$ can be calculated as
$$
p^*_i  = \theta_i  - log\left((1+\sum_{j\in S \backslash \{i\}} a_j) \times  W(\frac{exp(\theta_i -1) }{1+\sum_{j\in S \backslash \{i\}} a_j})\right),
$$
where $W(x)$ is the Lambert function and $a_j = exp(\theta_j-p_j)$ for all $j\in S$.
\end{lemma}
The proof of Lemma~\ref{lem:best_response} is in Appendix~\ref{app:sec:1}.  We further have the following result for such best response dynamics in the Bertrand  game.
\begin{lemma}\label{lemma:potential_game}
From an arbitrary price vector $\textbf{p}$,  the best response dynamics will converge to the Nash equilibrium of the Bertrand game in a finite number of steps.
\end{lemma}
The basic idea to derive this result is to show the Bertrand game is an ordinal potential game~\cite{MONDERER1996124} with a finite value; the detailed proof of Lemma~\ref{lemma:potential_game} is in Appendix~\ref{app:sec:2}.

\section{Optimal Segmenting Mechanisms}
In online marketplaces, the platform has control over search segmentation mechanisms - which set of products to display for a buyer. The platform can display any set of products, and the competition among selected sellers then takes place endogenously through the Bertrand game in Section~\ref{sec:oligopolistic}.
The goal of the platform is to decide the optimal products $S^* \subseteq \mathbb{S}$ to display, in order to maximize the equilibrium social welfare/revenue.
For $n$ potential products in the market, there are $2^n-1$ possible sets of products,
thus an exhaustive search to determine the optimal set of displayed products is infeasible.
We also note that the equilibrium constraint~(\ref{eqn:non_purchase}) imposed by the Bertrand competition game is highly nonlinear, which presents another challenge in deriving the optimal search segmentation mechanism.
 In this section, we exploit the structure of  social welfare/revenue functions to efficiently derive the optimal search segmentation mechanisms.

\subsection{Social Welfare Maximization}
In the following theorem, we show the online platform would display all products to maximize social welfare.
\begin{theorem}\label{theo:bertrand_sw}
For social welfare maximization, the optimal search segmentation mechanism is to display all products $\mathbb{S}$ in the platform.
\end{theorem}
\begin{proof}
We prove this theorem by showing that adding a new product will always improve the equilibrium social welfare. Suppose the platform has already selected sellers $S \subset \mathbb{S}$, and consider introducing a new product $j \in \mathbb{S} \backslash S$.
According to (\ref{nsw_bertrand}), we can express the equilibrium social welfare $\overline{sw}$ as
\begin{equation}\label{eqn:social_welfare1}
\overline{sw} = -\log{\bar{q}_0} +\sum_{i\in S} \frac{\bar{q}_i}{1-\bar{q}_i} + \frac{x_j\times \bar{q}_j}{1-x_j \times \bar{q}_j}.
\end{equation}
Here, $\bar{q}_i = V(\bar{q}_0 \times \exp(\theta_i -1))$ and $x_j$ is an indicator for product $j\in \mathbb{S}\backslash S$, where $x_j=1$ denotes product $j$ is selected for display; otherwise $x_j=0$. It is difficult to directly compare $\overline{sw}$ with $x_j=1$ and the one with $x_j=0$.
From~(\ref{eqn:non_purchase}), the demands ${\bm q}$ satisfy the following equilibrium constraint:
\begin{equation}\label{eqn:market_shares}
1-\bar{q}_0 =\sum_{i\in {S}} \bar{q}_i + x_j \times  \bar{q}_j = \sum_{i\in {S}} V(\bar{q}_0  \exp(\theta_i -1)) +x_j \times  V(\bar{q}_0  \exp(\theta_j -1)).
\end{equation}
Since $V(x)$ is an increasing function, we can observe from the above equation that $\bar{q}_0$ decreases when $x_j$ changes from $0$ to $1$.
Furthermore,  with (\ref{eqn:social_welfare1}) and (\ref{eqn:market_shares}), we can express the equilibrium social welfare as a function of $\bar{q}_0$:
\begin{equation}\label{eqn:sw_mechanism}
\overline{sw}(\bar{q}_0) = -\log{\bar{q}_0} +\sum_{i\in S} \frac{\bar{q}_i}{1-\bar{q}_i} + \frac{1-\bar{q}_0 - \sum_{i\in S} \bar{q}_i }{ \bar{q}_0+\sum_{i\in S} \bar{q}_i}.
\end{equation}
Thus, we only need to prove that $\overline{sw}(\bar{q}_0)$ is a decreasing function.
The  basic idea is to explicitly calculate the first derivative of $\overline{sw}(\bar{q}_0)$, and show $\overline{sw}'(\bar{q}_0)<0$.
We present the detailed proof of the following lemma in Appendix~\ref{app:sec:3}.
\begin{lemma}\label{lemma:sw_decreasing}
The social welfare $\overline{sw}(\bar{q}_0)$ is a decreasing function.
\end{lemma}
From this lemma and the above discussion,  we can always improve the equilibrium social welfare by adding a new product, which completes the proof.
\end{proof}


\subsection{Revenue Maximization}
The optimal search segmentation mechanism with the objective of revenue maximization is different from the optimal mechanism when the platform attempts to maximize social welfare. To illustrate this difference, we consider two cases:  a low quality case, \eg, $\theta_1 =\theta_2 = \cdots =\theta_n = 0.5$, and a high quality case,  \eg, $\theta_1=\theta_2 = \cdots =\theta_n = 10$.  From the result in Theorem~\ref{theo:bertrand_sw}, the optimal mechanisms  for social welfare maximization in  these two cases are to display all products. However, for  revenue maximization, it can be verified that the platform still displays all products in the low quality case, but only selects the first product in the high quality case.  We next show the design rationale for the optimal search segmentation mechanisms for the revenue maximization.

One critical decision the platform has to make is the following: given a set of products $S\subset \mathbb{S}$, whether to 
display the currently selected product set $S$, or add a new product $j$ from $\mathbb{S} \backslash S$.
We refer to such a decision problem as the ``incremental'' problem.
Similar to the discussion on social welfare maximization, given a set of selected products $S\subset \mathbb{S}$, we can represent the equilibrium revenue under these two decision options with the following function:
\begin{equation}\label{eqn:revenue_scenarios}
\overline{re} = \sum_{i\in S}\frac{\bar{q}_i}{1-\bar{q}_i}  + \frac{x_j \times \bar{q}_j}{1-x_j \times  \bar{q}_j}.
\end{equation}
We recall  that $x_j$ is an indicator for product $j\in \mathbb{S}\backslash S$, where $x_j=1$ indicates that product $j$ is selected for display; otherwise $x_j=0$. The  demands $\bar{q}_i$'s satisfy the following equilibrium constraint:
\begin{equation}\label{eqn:market_share_sum}
1-\bar{q}_0 = \sum_{i\in S} \bar{q}_i+x_j \times \bar{q}_j = \sum_{i\in S}  \bar{q}_i +x_j \times V(\bar{q}_0 \times  \exp(\theta_j-1)) .
\end{equation}
 Using  the above relation, we can further express the equilibrium revenue in   (\ref{eqn:revenue_scenarios}) as a function of $\bar{q}_0$:
 \begin{equation}\label{eqn:reve_equ}
 \overline{re}(\bar{q}_0) =\sum_{i\in S}\frac{\bar{q}_i}{1-\bar{q}_i} +  \frac{1}{\bar{q}_0+ \sum_{i\in S} \bar{q}_i } -1.
 \end{equation}
Since $V(x)$ is an increasing function, we have a critical observation from~(\ref{eqn:market_share_sum}): given a selected product set $S$, \emph{the quality $\theta_j$ of the potential product $j$ has a one-to-one and inverse relation with the demand  $\bar{q}_0$}, \ie, when $x_j=1$, involving the product with a higher quality $\theta_j$ leads to the lower value of $q_0$.
With this observation, we can derive the feasible range of the independent value $\bar{q}_0$.
On the one hand, when the platform selects the available product with the highest quality,  \ie, the product $j\in \mathbb{S} \backslash S$ with $\theta_j \geq  \theta_{j'}$ for all $j'\in \mathbb{S} \backslash S$, the demand $\bar{q}_0$ achieves its  lower bound at $\bar{q}_0^{min}$.
On the other hand, setting $x_j $ to $0$  represent the case that the platform does not select any new  product, and the corresponding  demand $\bar{q}_0^{max}$ in this case is the upper bound of $\bar{q}_0$.
Thus, we have  $\bar{q}_0 \in  \left[\bar{q}_0^{min}, \bar{q}_0^{max} \right]$ for all the  decision on selecting different possible product $j$.

\begin{figure}[!tbp]
  \centering
  \includegraphics[scale=0.8]{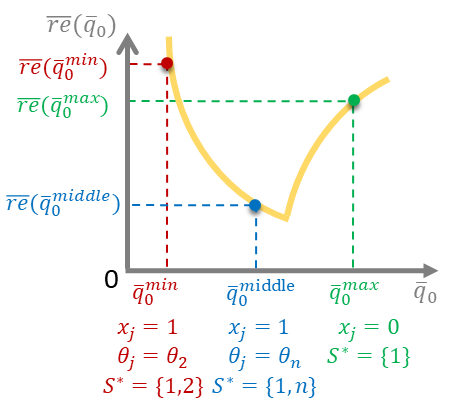}
  \caption{$\overline{re}(\bar{q}_0)$ is a quasi-convex revenue function for the possible product set to display when  the product $1$ has been selected.
  $\overline{re}(\bar{q}^{min}_0)$ is the revenue obtained by displaying $S^*= \{1,2\}$, $\overline{re}(\bar{q}^{middle}_0)$ is the revenue from showing $S^*= \{1,n\}$, and $\overline{re}(\bar{q}^{max}_0)$ is the revenue of displaying $S^*= \{1\}$.} \label{fig:quasi-convex}
\end{figure}

Based on the above discussion, we show that the optimal search segmentation mechanism is to choose the $k^*$ products with the best quality, for an appropriate value of $k^*$, using the following steps:
\begin{itemize}
\item First, we show that one should always display the product with the best quality to maximize revenue (Lemma~\ref{theo:s1}).
\item Then, we consider the decision of adding one product to  display. As discussed previously, we show that $re(q_0)$ is quasi-concave in $q_0$ which implies that the optimal decision is to add the next highest quality product or to not add a product at all. The quasi-convexity of $re(q_0)$ is shown in Lemma~\ref{lemma_quasi_convex_1} under a certain condition, which is relaxed in Appendix~\ref{app:sec:7}. Using the quasi-convexity of $re(q_0)$, in Lemma~\ref{theo:optimal_S=2}, we prove that if the optimal display set consists of $k^*$ products, then one should select the top $k^*$ products in terms of quality.
\item The final step is to find the optimal $k^*$. This can be done by the following calculation. For each possible value of $k^* \in \{2, \cdots, n\}$, we select the top $k^*$ products and find the revenue. We choose $k^*$ to maximize this revenue. This is clearly a linear-time algorithm in $n$, since one has to add one term to the expression for the revenue when we increase $k^*$ by one. This result is summarized in Theorem~\ref{theorem:revenue_bertrant_game}.
\end{itemize}

We first show that revenue maximization implies that the highest quality product is always selected for display.
 \begin{lemma}\label{theo:s1}
For revenue maximization,  it is optimal to always display the product with the highest quality.
 \end{lemma}
 The intuition behind the proof of this lemma  is to show that for any displayed product set, the revenue function in (\ref{eqn:reve_equ}) increases with the quality of the product with the highest quality in this set. The proof is in Appendix~\ref{app:sec:4}. 

Lemma~\ref{theo:s1} implies that when the optimal search segmentation mechanism is to display one product, \ie, $k^*=1$, the platform will choose the first product.
To obtain the result for the general case with $k^*\geq 2$, we need to establish the  quasi-convexity of the revenue function in~(\ref{eqn:reve_equ}).  It is non-trivial to directly verify this property
because the first term in the revenue function, \ie, $ \sum_{i\in S}\frac{\bar{q}_i}{1-\bar{q}_i}$, is increasing and concave  with respective to $q_0$, while the remaining term $ \frac{1}{\bar{q}_0+ \sum_{i\in S} \bar{q}_i } -1$ is decreasing and convex.
We first prove the desired quasi-convexity and design the optimal search segmentation mechanisms by assuming that all demands  $\bar{q}_i$'s are less than $0.5$, \ie, $q_1<0.5$, due to $q_i \leq q_1$ for all $i \in S$, meaning that no seller dominates the market.  This assumption simplifies the analysis, but still preserves the major intuition. Our results  also hold without this assumption, as shown in Appendix~\ref{app:sec:7}.
\begin{lemma}\label{lemma_quasi_convex_1}
For any selected product set $S$, the revenue function $\overline{re}(\bar{q}_0)$ in (\ref{eqn:reve_equ}) is quasi-convex in the range $\left[\bar{q}^{min}_0, \bar{q}^{max}_0\right]$, under the assumption of $q_1<0.5$.
\end{lemma}
The basic idea to prove this result is to check the second-order conditions of a quasi-convex function, \ie, at any point with zero slope, the second derivative is non-negative, \ie, $\overline{re}'(\bar{q}_0) = 0 \Rightarrow \overline{re}''(\bar{q}_0) > 0$. 
The details are in
Appendix~\ref{app:sec:5}. Equipped with Lemma~\ref{lemma_quasi_convex_1}, we can derive the optimal search mechanism for the case with $k^*\geq 2$.
\begin{lemma}\label{theo:optimal_S=2}
For revenue maximization, the optimal search segmentation mechanism is to display the top $k^*$ products if the cardinality of the optimal product set is $k^*\geq 2$, under the assumption of $q_1< 0.5$.
\end{lemma}
The optimality of the top $k^*$ mechanism in this Lemma can be established by showing that replacing any product with a product of higher quality will increase the revenue (see Appendix G for the proof).

While the specific value of $k^*$ depends on the quality of all products ${\bm \theta}$, The platform can find the optimal $k^*$ in linear time by computing the revenue of each set with the top $k\in [1,n]$ products, and selecting the one with the maximum revenue.  Thus, from Lemma~\ref{theo:s1} and  Lemma~\ref{theo:optimal_S=2},
we obtain the main result for the revenue maximization.
\begin{theorem}\label{theorem:revenue_bertrant_game}
For revenue maximization, the optimal search segmentation mechanism is to display the top  $k^*$ products, where $k^*$ is determined by the quality of all products ${\bm \theta}$, and can be calculated in linear time.
\end{theorem}

\section{Extensions to Cournot Competition Game}
In oligopolistic markets, another popular model to capture sellers' competition is Cournot game~\cite{daughety2005cournot}, in which sellers compete via controlling the supplies to products.
Specifically, each seller selects the number of units  she wants to sell, and the number of items that a seller wants to sell influences the availability of products and thus their prices. In other words, the prices are determined by the seller indirectly, by influencing supply. 
Although the Bertrand game appears to be more appropriate to capture the sellers' price competition in practice, such as the Airbnb or hotels.com case,  the Cournot game can also be used to model a specific type of price competition, \ie, a two-stage quantity precommitment price competition~\cite{doi:10.1287/opre.2018.1760,10.2307/3003636}. 
In this case, sellers compete on quantity in the first stage,  and then compete on price in the second stage with the fixed committed quantity. Sellers on Airbnb or hotels.com type platform can first compete via the number of units they sell, and then compete by the price per unit. 
It has been shown in~\cite{doi:10.1287/opre.2018.1760,10.2307/3003636} that under certain conditions, the equilibrium in Cournot  game is the equilibrium of such two-stage quantity precommitment price competition game. 
In this section, we discuss the existence and uniqueness of equilibrium, and then derive the optimal segmenting mechanisms for the Cournot game model. 

In a Cournot competition, seller $i \in {S}$ chooses a demand $q_i$ to maximize her revenue $r_i(\bm{q}) =p_i(\bm{q}) \times q_i $, where the price $p_i(\bm{q})$ of the product $i$ is determined by the demand vector $\bm{q}$ in (\ref{eqn:price}).
Thus, we express the revenue of sellers in terms of demand vector $\bm{q}$. Similar to the Bertrand game, we can represent a Cournot game as a triplet $G^c=\left({S}, (\mathcal{Q}_i)_{i\in {S}}, (r_i)_{i\in {S}}\right)$, where ${S}$ is a set of players, $\mathcal{Q}_i$ is the strategy space of player $i\in {S}$, and $r_i(\bm{q})$ is the payoff of player $i\in {S}$.
We represent the set of strategy profiles $\mathcal{Q}=\mathcal{Q}_1\times \mathcal{Q}_2 \times \cdots \times \mathcal{Q}_n$.
According to price-demand relation in (\ref{demand}), the requirement of non-negative prices implies the feasible strategy space
\begin{equation*}
\mathcal{Q}=\left\{\bm{q}:0  \leq \sum_{i\in \bar{S}} q_i \leq \frac{\sum_{i\in \bar{S}}\exp(\theta_i)}{1+\sum_{i\in \bar{S}}\exp(\theta_i)},  \forall \bar{S} \subseteq  {S}. \right\},
\end{equation*}
which is convex and compact.
We also denote a feasible strategy profile $\bm{q}\in \mathcal{Q}$ as $\bm{q}=(q_i, \bm{q}_{-i})$. Using the facts that the payoff function $r_i(q_i, \bm{q}_{-i})$ is a strictly concave function with respective to $q_i$ and the feasible strategy profile space ${\bm Q}$ is convex and compact, the Cournot game $G^c$ is a concave game as defined in Rosen's paper~\cite{10.2307/1911749}.
From Rosen's result~\cite{10.2307/1911749}, we know that there exists a unique (pure) Nash equilibrium, and the Nash equilibrium can be obtained from
the system of first-order-condition equations.

Setting the partial derivative $\partial r_i(q_i, \bm{q}_{-i})/ \partial q_i$ to be zero for all $i \in {S}$, we have
$$
\frac{\partial r_i(q_i, \bm{q}_{-i})}{\partial q_i}=\theta_i -1 -\frac{q_i}{1-\sum_{j \in {S}} q_j} -\log\left(\frac{q_i}{1-\sum_{j \in {S}} q_j}\right)=0, \  \forall i \in {S}.
$$
We define a variable $w_i \triangleq q_i/(1-\sum_{j \in {S}} q_j)$, and the above equations become
$$
\exp(\theta_i-1) = w_i \times \exp(w_i), \quad \forall i \in {S}.
$$
Thus, we  get $w_i = W(\exp(\theta_i-1))$, where $W(x)$ is the Lambert function. According to the definition of $w_i$, we  obtain the equilibrium demands 
$$
\hat{q}_i = \frac{w_i}{1+\sum_{j \in {S}} w_j} = \frac{W(\exp(\theta_i-1))}{1+\sum_{j=1}^{n} W(\exp(\theta_j-1))}, \quad \forall i \in {S}.
$$
We can verify that such a demand vector $\hat{\bm{q}}$ satisfies the feasible constraint, \ie, $\hat{\bm{q}}\in {\bm Q}$. Substituting $\hat{\bm{q}}$ into (\ref{eqn:price}), we obtain the corresponding equilibrium prices
$$
\hat{p}_i = 1+ W(\exp(\theta_i-1)).
$$
We summarize the Nash equilibrium of the Cournot competition game in the following theorem.
\begin{theorem}
There exists a unique (pure) Nash equilibrium in the Cournot game $G^c=\left({S}, (\mathcal{Q}_i)_{i\in {S}}, (r_i)_{i\in {S}}\right)$. The Nash equilibrium demand $\hat{q}_i$ and price $\hat{p}_i$ for each seller $i\in {S}$ are given by
$$
\hat{q}_i = \frac{W(\exp(\theta_i-1))}{1+ \sum_{j \in {S}} W(\exp(\theta_j-1))},  \quad \hat{p}_i= 1+W(\exp(\theta_i-1)).
$$
\end{theorem}
Substituting equilibrium solutions into (\ref{eqn:sw_price}), we can get the equilibrium social welfare in the Cournot game with the seller set $S\subseteq \mathbb{S}$
\begin{equation}\label{nsw_cournot}
\widehat{sw}(S) = \log(1+\sum_{i \in {S}} w_i ) +\frac{\sum_{i \in {S}} (w_i^2 +w_i)}{1+\sum_{i \in {S}} w_i}.
\end{equation}
Similarly, for the set of sellers $S\subseteq \mathbb{S}$, the equilibrium revenue in the Cournot game is
\begin{equation}\label{nre_cournot}
\widehat{re}(S) = \frac{\sum_{i \in {S}} (w_i^2 +w_i)}{1+\sum_{i \in {S}} w_i}.
\end{equation}
\subsection{Social Welfare Maximization}
In contrast to the Bertrand game, it may be possible for the platform to  achieve higher equilibrium social welfare by just displaying  a subset of products in the Cournot game.
We construct a simple instance to illustrate this difference. Suppose there is only one product with high quality, and the remaining products have low quality, \eg, $\theta_1=10$ and $\theta_i=0$ for all $i\in \mathbb{S} \backslash \{1\}$. By the equilibrium social welfare in (\ref{nsw_cournot}), we can verify that the optimal search segmentation mechanism in the Cournot game is just to select the first product, while the optimal mechanism in the Bertrand game is to display all products according to the result in Theorem~\ref{theo:bertrand_sw}.

Before presenting the main result, we first show an important lemma for the Cournot game.
\begin{lemma}\label{lem:social_cournot}
For social welfare maximization in a Cournot game, the optimal search segmentation mechanism is to display the top $k$ products if the cardinality of the optimal product set is $k$.
\end{lemma}
As in the case of revenue maximization in Bertrand competition, the idea to establish this result is also to make a connection of quasi-convexity of equilibrium social welfare/revenue functions with the decision on displaying products. The detailed proof of this lemma is in Appendix~\ref{app:sec:8}. 
 
We have shown that to maximize equilibrium social welfare, the platform may not display all products in a Cournot game.
To find the optimal number of products $k^*$,
 which depends on the product quality vector ${\bm \theta}$, the platform can calculate equilibrium social welfare for each possible $k$, and select the one with the maximum social welfare. By Lemma~\ref{lem:social_cournot}, for each candidate $k$, we only need to consider the set containing the top $k$ products. Thus, the platform can determine the optimal $k^*$ in linear time, leading to the
following main result.
\begin{theorem}\label{theorem:cournot_game}
For social welfare maximization, when sellers compete via a Cournot game, the optimal search segmentation mechanism is to display the top $k^*$ products, where $k^*$ is determined by the quality of all products ${\bm \theta}$, and can be calculated in linear time.
\end{theorem}

\subsection{Revenue Maximization}
Similar to social welfare maximization, the platform also displays a subset of products to maximize equilibrium revenue in a Cournot game. As before, we first present a useful lemma.
\begin{lemma}\label{lem_revenue_cournot}
For revenue maximization in the Cournot game, the optimal search segmentation mechanism is to display the top $k$ products if the cardinality of the optimal product set is $k$.
\end{lemma}
The key idea to prove this lemma directly follows from the proof for Lemma~\ref{lem:social_cournot}, and we defer it to 
the Appendix~\ref{app:sec:9}.  With this lemma, we can derive the following result for revenue maximization in the Cournot game.
\begin{theorem}\label{theorem:revenue_cournot_game}
For revenue maximization, when sellers compete via a Cournot game, the optimal search segmentation mechanism is to involve the top $k^*$ products, where $k^*$ is determined by the quality of all products ${\bm \theta}$, and can be calculated in linear time.
\end{theorem}

 \section{Related Work}
Our work is related to the burgeoning literature that studies online platform marketplaces of using control levels other than pricing to influence the market outcomes~\cite{Arnosti:2014:MCD:2600057.2602893,Banerjee:2017:STM:3038912.3052578,Birge:2018:OCS:3219166.3219216,kanoria2017facilitating}.
 Kanoria and Saban designed a framework to facilitate the search for buyers and sellers on matching platforms, and found that simple restrictions on what buyers/sellers can access would boost social welfare~\cite{kanoria2017facilitating}. Arnosti~\et~investigated the welfare loss due to the uncertainty about seller availability in asynchronous dynamic matching markets, and also found that limiting the visibility of sellers can improve social welfare~\cite{Arnosti:2014:MCD:2600057.2602893}.
Our result, displaying only a subset of  products to buyers can increase the equilibrium revenue, extends the  findings in these two pieces of work to the context of revenue optimization.
Banerjee~\et~studied how the platform should control which sellers and buyers are visible to each other, and provided polynomial-time approximation algorithms to optimize social welfare and throughput~\cite{Banerjee:2017:STM:3038912.3052578}. In their model, supply and demand are associated with public distributions.
By contrast, we adopt the MNL model to derive a specific demand system, and use the Bertrand game to capture  supply response to this demand system, doing so leads to very different optimization problems.
There are also other types of control mechanisms that can be used to efficiently operate the online platforms, such as
commission rates and subscription fees~\cite{Birge:2018:OCS:3219166.3219216} and prices and wages for buyers and sellers~\cite{47759}.

Revenue management under the MNL and general demand model has been extensively studied in economics, marketing and operation management~\cite{doi:10.1287/msom.1080.0221,doi:10.1111/poms.12191,song2007demand,doi:10.1287/mnsc.1030.0147}.
The model considered in this paper is closely related to that in assortment optimization, which is an active area in revenue management research. In the problem of assortment optimization,
the demand of products are governed by the variants of attraction-based choice models, such as  MNL model~\cite{McFa73},  mixed nested logit model~\cite{CARDELL1980423} and nested logit model~\cite{doi:10.1068/a090285}, and each product is associated with a fixed price.
The objective is to find a set of products, or an assortment to offer that maximizes the expected revenue. In~\cite{doi:10.1287/mnsc.1030.0147}, Talluri and van Ryzin studied the assortment optimization problem under the MNL model, and showed that the optimal assortment includes a certain number of products with the highest prices.
We also derive a similar result, but use the criteria of quality rather than price to rank the potential products.
In our setting a key difference from this line of work is that the product prices are determined endogenously by the outcome of oligopolistic  competition games instead of being given beforehand.
Pricing multiple differentiated products in the context of the MNL model is another fairly active direction in  literature~\cite{doi:10.1287/msom.1080.0221,doi:10.1287/mnsc.42.7.992,song2007demand}.
Different from the assortment optimization problem, in this setting, all the products are displayed, and the objective is to choose pries for the products to maximize revenue.
In contrast, we focus on search segmentation mechanisms with endogenous prices, where the platform only controls the set of displayed products, to optimize the equilibrium social welfare/revenue.

Bertrand competition, proposed by Joseph Bertrand in 1883, and Cournot competition, introduced in 1838 by Antoine Augustin Cournot, are fundamental economic models that represent sellers competing in a single market, and have been studied comprehensively in economics.
Due to the motivation that many sellers compete in more than one market in modern dynamic and diverse economy, a recent and growing literature has studied Cournot competitions in network environments~\cite{10.1007/978-3-319-13129,Bimpikis:2014:CCN:2600057.2602882,8264340,8057125}. 
The work \cite{10.1007/978-3-319-13129,Bimpikis:2014:CCN:2600057.2602882} focused on characterizing and computing Nash equilibria, and investigated the impact of changes in the (bipartite) network structure on seller's profit and buyer's surplus.
\cite{8264340} and \cite{8057125} analyzed the efficiency loss of networked Cournot competition game via the metric of price of anarchy.
While all these previous works focused on the objective of social welfare maximization in networked Cournot competition, we consider the objective of both social welfare and revenue maximization in the networked 
Bertrand competition, and also in the networked Cournot competition. We further provide efficient segmenting mechanisms to optimize the social welfare/revenue under the Nash equilibrium.
 \section{Conclusion and Future Work}
 In this paper, we have studied the problems of social welfare maximization and revenue maximization in designing search space for online platform markets.
In the discriminatory control model, the platform can only control the search segmentation mechanisms, \ie, determine the list of products to display for buyers, and the products' prices are determined endogenously by the competition among sellers.
 Under the standard buyer choice model, namely the multinomial logit mode,
we have developed efficient and optimal search segmentation mechanisms to maximize the equilibrium social welfare and revenue under Bertrand competition game.
For social welfare maximization, it is optimal to display all the products. 
For revenue maximization, the optimal search mechanism, referred as quality-order mechanism, is to display the top $k^*$ highest quality products, where $k^*$ can be computed  in at most linear time in the number of products.
We extend our results to Cournot competition game, and show that the optimal search segmentation mechanisms are also the simple quality-order mechanisms, for the objectives of both social welfare and revenue maximization. 

 One possible direction for future work is to extend the quality-order mechanisms to more complex demand models (such as  mixed MNL model and  nested logit model) and general buyer-seller (bipartite) networks. Another interesting research topic is to design the optimal search mechanisms in dynamic setting, where buyers arrive and depart, and sellers have limited capacity for products.

\bibliographystyle{abbrv}
\bibliography{sigmetrics}

\appendix

\section{Monopolistic Market}\label{sec:monopolistic_market}
In this section, we consider an online marketplace in the full control model as a monopolistic market, where the platform (or monopoly) jointly determines the price vector $\bm{p}$ and the set of displayed products to maximize social welfare or revenue. We first fix the displayed products as $S \subseteq \mathbb{S}$, and investigate the optimal monopoly prices.

It turns out that social welfare function in (\ref{eqn:sw_price}) and revenue function in (\ref{eqn:revenue_price})  are not concave in prices~\cite{doi:10.1287/mnsc.42.7.992}. The key observation here is that the social welfare and revenue functions become jointly concave if we work with the demand vector~\cite{doi:10.1287/msom.1080.0221,song2007demand}. Using the price-demand relation in (\ref{eqn:price}), we can express the social welfare  as a function of demand vector $\boldsymbol{q}$:
\begin{equation}\label{eqn:sw_market_share}
sw(\boldsymbol{q})= -(1-\sum_{j \in {S}} q_j)\times \log(1-\sum_{j \in {S}} q_j) + \sum_{i \in {S}}(\theta_i -\log(q_i))\times q_i.
\end{equation}
Similarly, the revenue function can be re-written as
\begin{equation}\label{eqn:re_market_share}
{re}({\bm q}) = \sum_{i \in {S}} (\theta_i  + \log(1-\sum_{j \in {S}} q_j )   -\log(q_i) ) \times q_i.
\end{equation}
According to (\ref{demand}), the requirement of non-negative prices implies that the feasible demand space is given by 
\begin{equation}\label{eqn:feasible_q}
\mathcal{Q}=\left\{\bm{q}:0  \leq \sum_{i\in S} q_i \leq \frac{\sum_{i\in S}\exp(\theta_i)}{1+\sum_{i\in S}\exp(\theta_i)},  \forall S \subseteq  \mathbb{S}. \right\},
\end{equation}
which is convex and compact.
Therefore, the problem of social welfare maximization in the monopolistic market can be formulated as
\begin{eqnarray}
\nonumber
 \max_{\bm{q}\in \mathcal{Q}}  \quad  sw(\boldsymbol{q})= -(1-\sum_{j \in {S}} q_j ) \log(1-\sum_{j \in {S}} q_j ) + \sum_{i \in {S}}(\theta_i -\log(q_i)) q_i.
\end{eqnarray}
 This optimization problem is  a standard concave maximization problem. By Karush-Kuhn-Tucker (KKT) condition~\cite{boyd2004convex}, we set partial derivative $\frac{ \partial sw(\boldsymbol{q}) }{  \partial q_i } =0 $, and re-formulate this equation to get
 $$
 exp(\theta_i) = \frac{q^*_i }{1-\sum_{j \in {S}} q^*_j},  \quad \quad \forall i \in {S}.
 $$
Solving this system of equations, we can derive the optimal market shares $\bm{q}^*$ and the corresponding optimal prices $\bm{p}^*$ for social welfare maximization in the monopolistic market:
$$
q_i^* = \frac{\exp(\theta_i)}{1+ \sum_{j \in {S}} \exp(\theta_j)}, \quad p_i^*=0, \quad \forall i\in {S}.
$$
In the optimal solution for social welfare maximization, the platform sets prices of all products to be zero, and the demands are proportional to their product quality. The revenue of the seller is zero, and the utility of the buyer is maximized. The optimal social welfare for displaying products $S \subseteq \mathbb{S}$ is
\begin{equation}\label{eqn:sw_monopoly}
sw^*(S)=\log(1 + \sum_{i \in {S}} \exp(\theta_i)).
\end{equation}

Similarly, revenue maximization in the monopolistic market is also a convex optimization problem.
can be formulated as
\begin{eqnarray}
\nonumber
 \max_{\bm{q}\in \mathcal{Q}}  \quad  {re}(\boldsymbol{q})= \sum_{i \in {S}} (\theta_i  + \log(1-\sum_{j \in {S}} q_j )   -\log(q_i) ) \times q_i.
\end{eqnarray}
By setting the partial derivative $\frac{ \partial {re}(\boldsymbol{q}) }{  \partial q_i } =0$, we can get
$$
\theta_i-1  = \log ( \frac{q^*_i}{1-\sum_{j \in {S}} q^*_j}  ) + \frac{\sum_{j \in {S}} q^*_j}{1-\sum_{j \in {S}} q^*_j}  \quad \quad  i \in {S}.
$$
Using the relation between $q_i$ and $p_i$ in (\ref{demand}), we can express the above equations  in terms of $p_i$'s to simplify the calculation:
\begin{equation}\label{eqn:derivative_system}
exp(\theta_i -1)  =  a_i \times exp(\sum_{j \in {S}} a_j ), \quad \quad  i \in {S}.
\end{equation}
Solving this system of equations, we can obtain
\begin{equation}\label{eqn:sum_a}
\sum_{j \in {S}} a_j = W(\sum_{j \in {S}}  exp(\theta_j -1)),
\end{equation}
where $W(x)$ is the Lambert function~\cite{Corless1996}, and is the solution $w$ satisfying
$
w\times \exp(w) = x.
$
Substituting (\ref{eqn:sum_a}) back into (\ref{eqn:derivative_system}), we can calculate the optimal price $p^*$ and demand $q^*$ for revenue maximization in the monopolistic market
$$
 p_i^*  =1 + \omega,  \quad q_i^*  =    \frac{\exp\left( ( \theta_i-1 )- \omega  \right)}{1+ \omega}, \quad \forall i\in {S},
$$
where $\omega \triangleq W(\sum_{j \in {S}}  exp(\theta_j -1))$.
In the optimal solution for revenue maximization, the platform sets the same price for all products, and the demands are also proportional to the product quality. The optimal revenue for the set of selected sellers $S \subseteq \mathbb{S}$ is
\begin{equation}\label{eqn:re_monopoly}
re^*({S}) = W(\sum_{j \in {S}}  exp(\theta_j -1)  ).
\end{equation}
From (\ref{eqn:sw_monopoly}) and (\ref{eqn:re_monopoly}), we can observe that both the optimal social welfare and  revenue in the monopolistic market are increasing with respective to the number of displayed products. Thus, in the  full control model, the platform displays all products $\mathbb{S}$ to buyers, and sets the optimal prices ${\bm p^*}$ to maximize the social welfare or revenue.
\section{Proof for Lemma~\ref{lem:best_response}}\label{app:sec:1}
\begin{proof}
From the relation between price and demand in (\ref{eqn:price}), we can write the seller $i$'s revenue in (\ref{eqn:seller_revenue}) with respective to the demand $q_i$:
\begin{equation}\label{eqn:revenue_deamnd}
r_i(q_i) = q_i  \times (\theta_i + \log(q_0) - \log(q_i)).
\end{equation}
Here, we have used $q_0=1- \sum_{j \in S} q_j$. The variable $q_0$ depends on all $q_j$'s, which makes it difficult to calculate the optimal demand $q^*_i$. We next express $q_0$ only using $q_i$.
From (\ref{demand}), we have  $q_0 \times a_j =q_j,  \forall j\in S$.
We summarize these equations over all $j\in S \backslash \{i\}$, and obtain
$
q_0 \times \sum_{j\in S \backslash \{i\}} a_j = \sum_{j\in S \backslash \{i\}} q_j.
$
Combining with $q_0 = 1- \sum_{j \in S} q_j$, we have
\begin{equation}\label{eqn:q_0}
q_0  = \frac{1-q_i}{1+\sum_{j\in S \backslash \{i\}} a_j }.
\end{equation}
We note that given a vector of fixed prices $\textbf{p}_{-i}$, the $a_j=exp(\theta_j -p_j), \forall j\in S \backslash\{i\} $ are constants. Using (\ref{eqn:q_0}), we can rewrite the revenue function $r_i(q_i)$ as
$$
r_i(q_i) = q_i  \times \left(\theta_i   + \log(\frac{q_i}{1-q_i}) -  \log(1+\sum_{j\in S \backslash \{i\}} a_j )\right),
$$
and that it is strictly concave.
We calculate the first derivative of $r_i(q_i)$, and set it to be zero:
\begin{eqnarray}
\nonumber
& & r'_i(q_i) = \theta_i  - \log(\frac{q_i}{1-q_i}) -  \log(1+\sum_{j\in S \backslash \{i\}} a_j ) -\frac{q_i}{1-q_i} -1, \\
\nonumber
& & r'_i(q^*_i) =  0   \Rightarrow  (\frac{q^*_i}{1-q^*_i}) \times exp(\frac{q^*_i}{1-q^*_i}) =  \frac{exp(\theta_i -1) }{1+\sum_{j\in S \backslash \{i\}} a_j}, \\
& \Rightarrow &  \frac{q^*_i}{1-q^*_i}  =  W(\frac{exp(\theta_i -1) }{1+\sum_{j\in S \backslash \{i\}} a_j}),\label{eqn:equilibrium}
\end{eqnarray}
where $W(x)$ is the Lambert function.
With $q_0 \times a^*_i= q^*_i$ and (\ref{eqn:q_0}), we have
\begin{equation*}
 \frac{a^*_i}{1+\sum_{j\in S \backslash \{i\}} a_j }  = \frac{q^*_i}{1-q^*_i}
\end{equation*}
Together with  (\ref{eqn:equilibrium}), we can get
\begin{eqnarray*}
& &  a^*_i =    W(\frac{exp(\theta_i -1) }{1+\sum_{j\in S \backslash \{i\}} a_j}) \times (1+\sum_{j\in S \backslash \{i\}} a_j) \\
& \Rightarrow  &  p^*_i  = \theta_i  - log\left( W(\frac{exp(\theta_i -1) }{1+\sum_{j\in S \backslash \{i\}} a_j}) \times (1+\sum_{j\in S \backslash \{i\}} a_j) \right).
\end{eqnarray*}
We need to guarantee that the new price $p^*_i$ is non-negative. From the definition of Lambert W function, we have $W(x) \leq x$, and thus we can derive from the above equation that $p^*_i \geq 1$. Since  the other prices $\bm{p}_{-i}$ remain the same and are non-negative, the new price vector $(p^*_i,\bm{p}_{-i})$ is feasible. 
\end{proof}
\section{Proof for Lemma~\ref{lemma:potential_game}}\label{app:sec:2}
\begin{proof}
We first show that the Bertrand game is an ordinal potential game. We construct a potential function
$$
G(\textbf{p}) = \frac{\prod_{i\in S} p_i \times a_i }{1+ \sum_{j\in S} a_j},
$$
which satisfies the following property for every price vector $\textbf{p}$, every seller $i \in S$ and every unilateral deviation $p'$ by $i$:
\begin{equation}\label{eqn:potential_function}
r_i(p_i, \textbf{p}_{-i}) > r_i(p'_i, \textbf{p}_{-i})  \Leftrightarrow G(p_i, \textbf{p}_{-i}) > G(p'_i, \textbf{p}_{-i}).
\end{equation}
Since the revenue function $r_i(q_i)$ in (\ref{eqn:revenue_deamnd}) is strictly concave, the deviating seller $i$'s revenue strictly increases. By (\ref{eqn:potential_function}), the potential function also strictly increases  after each iteration of the best response dynamics. Thus, no cycles are possible.
Since the potential function has a finite value, the best response dynamics eventually reach the maxima of the potential function, \ie, the Nash equilibrium, in finite steps.
\end{proof}
\section{Proof for Lemma~\ref{lemma:sw_decreasing}}\label{app:sec:3}
\begin{proof}
By implicit differentiation of $V(x)$ in (\ref{eqn:v_func}), we can calculate the derivative of $V(x)$:
$$
V'(x) = \frac{1}{x\times \left(\frac{1}{V(x)}+\frac{1}{(1-V(x))^2}\right)}.
$$
We can then get the derivative of $\overline{sw}(\bar{q}_0)$ in (\ref{eqn:sw_mechanism})
$$
\overline{sw}'(\bar{q}_0) = -\frac{1}{\bar{q}_0} + \sum_{i\in S} \frac{\bar{q}'_i}{(1-\bar{q}_i)^2} - \frac{1+\sum_{i\in S} \bar{q}'_i}{ \bar{q}_0 +\sum_{i\in S} \bar{q}_i},
$$
where $\bar{q}'_i = \exp(\theta_i-1) \times V'(\bar{q}_0\times \exp(\theta_i-1))$.
Substituting the specific form of $V'(x)$ into the above equation, we further have
\begin{eqnarray}
\nonumber
\overline{sw}'(\bar{q}_0) & < & -\frac{1}{\bar{q}_0} + \sum_{i\in S} \frac{\bar{q}'_i}{(1-\bar{q}_i)^2} - \sum_{i\in S} \bar{q}'_i  \\
\nonumber
&  = &  -\frac{1}{\bar{q}_0} + \sum_{i\in S} \bar{q}'_i \times \frac{1-(1-\bar{q}_i)^2}{(1-\bar{q}_i)^2}\\
& = & \frac{1}{\bar{q}_0} \left( -1 + \sum_{i\in S} \frac{1-(1-\bar{q}_i)^2}{\frac{(1-\bar{q}_i)^2}{\bar{q}_i}+1} \right). \label{eqn:derivative}
\end{eqnarray}
The first inequality is due to $0 \leq \bar{q}_0 + \sum_{i\in S} \bar{q}_i \leq 1$.
When $S=\varnothing$, $\overline{sw}'(\bar{q}_0)$ is negative. It is easy to check that $\overline{sw}'(\bar{q}_0)$ is also negative when $|S|=1$.
Thus, without loss of generality, we can assume that at least two products have been selected, \ie, $|S|\geq 2$.
We introduce an increasing sigmoidal function
$$f(q) \triangleq \frac{1-(1-q)^2}{\frac{(1-q)^2}{q}+1},$$
which is convex in range $[0,0.5]$ and concave in range $[0.5, 1]$. We plot the graph of $f(q)$ in Figure~\ref{fig:functionf}. We also plot the graphs of linear functions $f_1(q)=q$ and $f_2(q)=a\times q$, where the slope is $a= \frac{f(1-q_1)}{1-q_1}$ with $q_1 \geq 0.5$. We summarize two important properties of the function $f(q)$.

\begin{figure}[!tbp]
  \centering
  \includegraphics[scale=0.5]{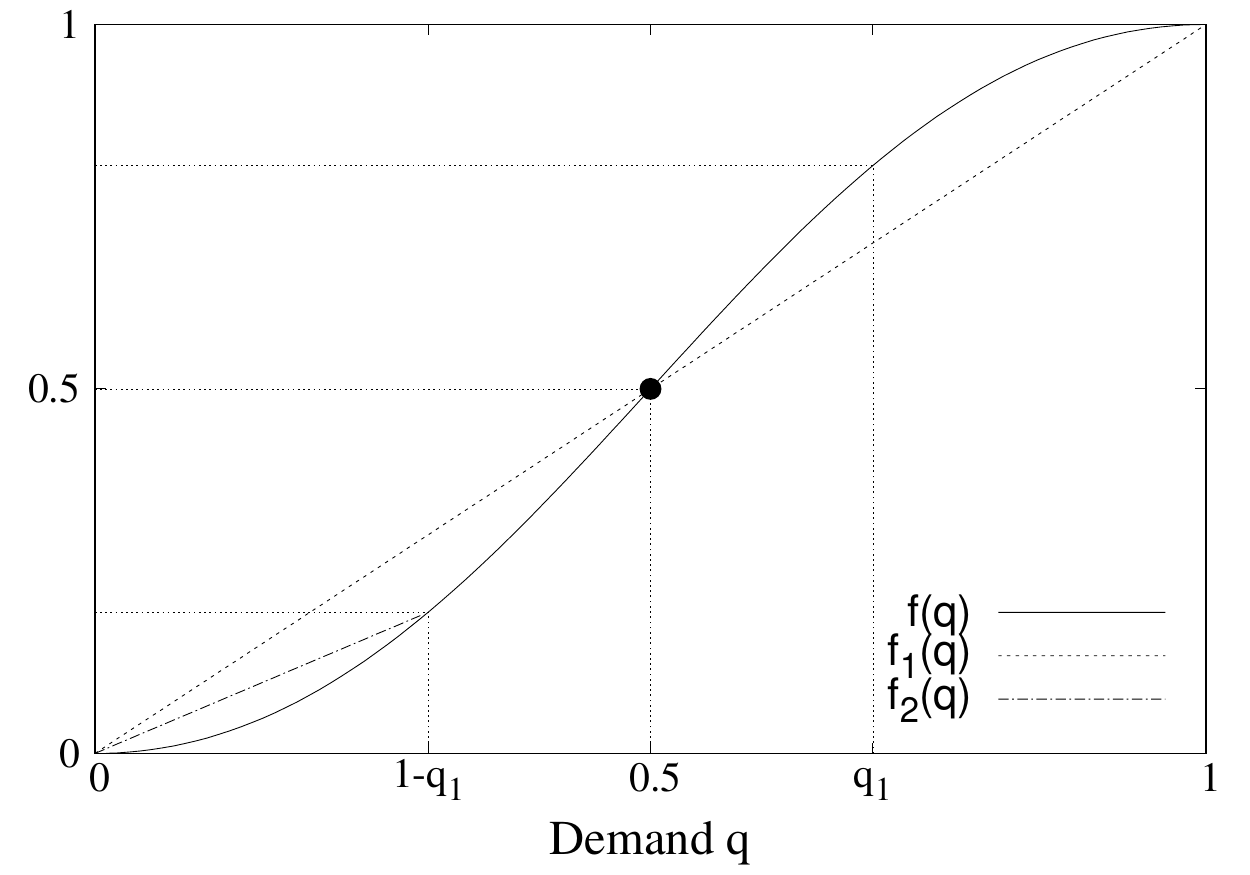}
  \caption{The function $f(q)$ is sigmoidal over the interval $[0,1]$, and is convex in $[0,0.5]$ and concave in $[0.5,1]$.}
  \label{fig:functionf}
\end{figure}

$\bullet$ For any $q\in [0,1]$, $f(q)+f(1-q)=1$.

$\bullet$ For any $q\in [0,0.5]$, $f(q)\leq f_1(q) $, and $f(q)\leq f_2(q)$ for any $q\in [0, 1- q_1]$ with $q_1\geq 0.5$.

To prove (\ref{eqn:derivative}) is non-positive, we consider the following optimization problem:
$$ \text{Maximize}   \sum_{i\in S} f(\bar{q}_i), \ \ \ \ \text{Subject to}  \  \sum_{i\in S} \bar{q}_i=1,$$
and show the optimal objecitve is always no more than $1$, \ie, $\sum_{i\in S} f(\bar{q}_i)\leq 1$ by separately considering the following two cases:

$\bullet$ When $|S|=2$, according to the first property of function $f(q)$, we have  $\sum_{i\in S} f(q_i) = 1$.

$\bullet$ When $|S|>2$, we further consider two different scenarios.
If $\bar{q}_i\leq 0.5$ for all $i\in S$, by the second property of $f(q)$, we have
$$
\sum_{i\in S} f(\bar{q}_i) \leq \sum_{i\in S} \bar{q}_i =1.
$$
In the other case, one of the $\bar{q}_i$'s is larger than $0.5$. Since $\sum_{i\in S }\bar{q}_i =1$ and $0 \leq \bar{q}_i \leq \bar{q}_1$,  we can claim $\bar{q}_1>0.5$ and $\bar{q}_i <  1 - \bar{q}_1$ for all $i\in S\backslash \{1\}$. By the second property of $f(q)$, we have
$$
f(\bar{q}_i) \leq f_2(\bar{q}_i) = \frac{f(1-\bar{q}_1)}{1-\bar{q}_1} \times \bar{q}_i, \quad \forall i \in S\backslash\{1\}.
$$
Thus, we can further get
\begin{eqnarray*}
f(\bar{q}_1) + \sum_{i \in S\backslash\{1\}} f(\bar{q}_i) & \leq & f(\bar{q}_1) +  \frac{f(1-\bar{q}_1)}{1-\bar{q}_1} \times \sum_{i \in S\backslash\{1\}} \bar{q}_i \\
&  = &  f(\bar{q}_1) +  f(1-\bar{q}_1) =  1.
\end{eqnarray*}
The first equality comes from $\sum_{i\in S} \bar{q}_i=1$, and the second equality is due to the first property of $f(q)$.

From the above discussion, we have proved that the maximum value of $\sum_{i\in S} f(\bar{q}_i)$ is no more than $1$. By (\ref{eqn:derivative}), it follows that $\overline{sw}'(\bar{q}_0) < 0$,
\end{proof}
\section{Proof for Lemma~\ref{theo:s1}}\label{app:sec:4}
 \begin{proof}
 Suppose the displayed product set $S$ does not contain the product with the highest quality.
We denote the product with the highest quality in $S$ as $i^*$, \ie, $ \bar{q}_{i^*} \geq \bar{q}_i, \forall i\in S$. We only need to show that the revenue does not decrease if we replace  the product $i^*$ with the product $1$, which is equivalent to increase the quality of product $i^*$ from $\theta_{i^*}$ to $\theta_{1}$.
From (\ref{eqn:reve_equ}), we can calculate the derivative of the revenue function with the product set  $S$:
\begin{equation}\label{eq_kl}
 \overline{re}'(\bar{q}_0)  = \sum_{i\in S \backslash i^* }\frac{\bar{q}'_i}{(1-\bar{q}_i)^2}-\frac{1+ \sum_{i\in S \backslash i^* }\bar{q}'_i}{(1-\bar{q}_{i^*})^2}.
\end{equation}
Here, we use the equilibrium constraint $\bar{q}_{i^*}+ \sum_{i\in S \backslash i^* } \bar{q}_i =1- \bar{q}_0$.
Since $ \bar{q}_{i^*} \geq \bar{q}_i, \forall i  \in S \backslash i^*$,
we can derive that  $\overline{re}'(\bar{q}_0) \leq 0$, and thus $\overline{re}(\bar{q}_0)$ is a non-increasing function.
Furthermore, as the quality $\theta_{i^*}$ has an inverse relation with the demand  $\bar{q}_0$ from (\ref{eqn:market_share_sum}),
increasing $\theta_{i^*}$ to $\theta_{1}$ is equivalent to decrease $\bar{q}_0$, which would not decrease the revenue.
From the above discussion, we can obtain the result.
 \end{proof}

\section{Proof for Lemma~\ref{lemma_quasi_convex_1}}\label{app:sec:5}
\begin{figure}[!tbp]
  \centering
  \includegraphics[scale=0.5]{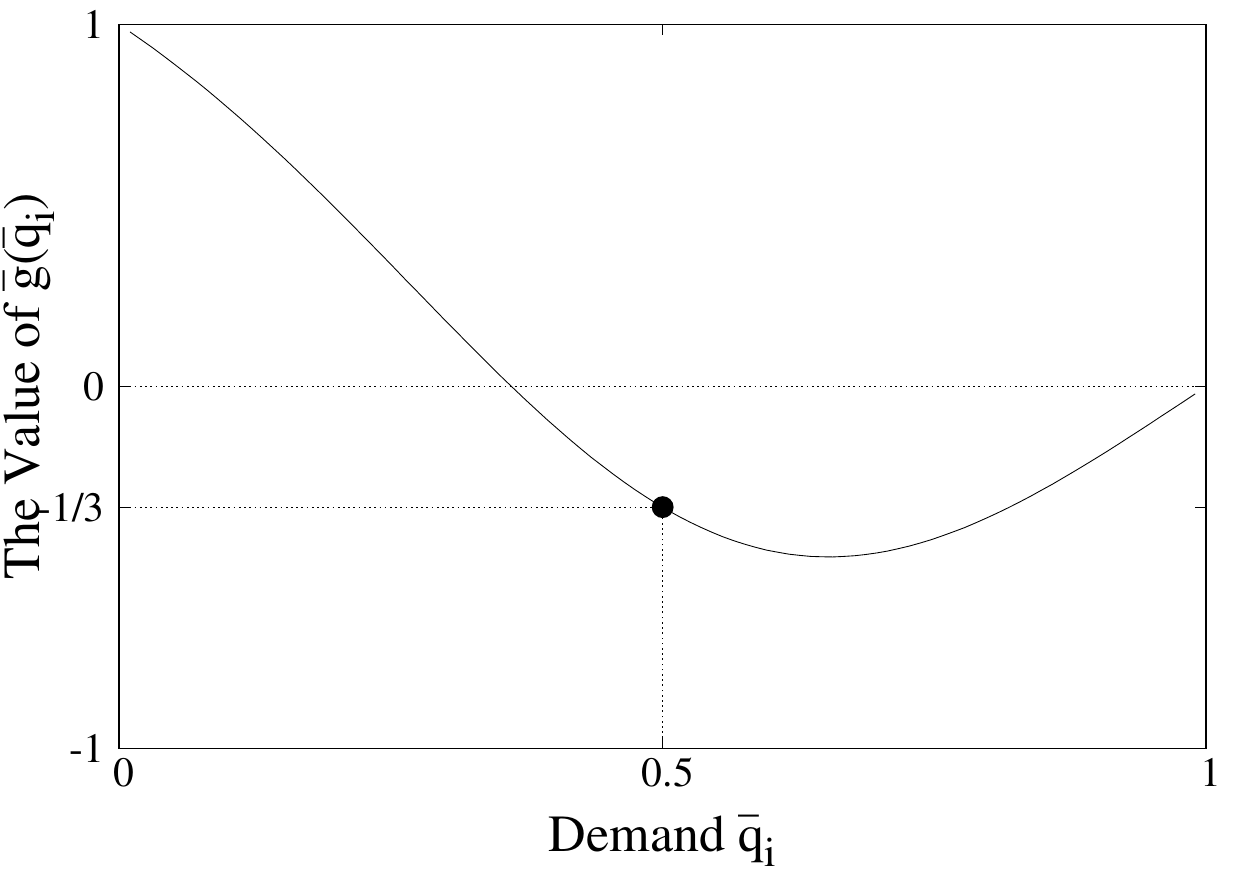}
  \caption{The function $g(\bar{q})$ has one critical property: for any pair of $\bar{q}_i$ and $\bar{q}_j$ with $0.5 \geq \bar{q}_i \geq \bar{q}_j\geq 0$, the relation $-1/3 \leq g(\bar{q}_i) \leq g(\bar{q}_j ) \leq 1$ holds.}
  \label{fig:functiong}
\end{figure}

\begin{proof}
The basic idea can be illustrated  via the case of $k^* = 2$ and $S=\{1\}$.  From~(\ref{eqn:reve_equ}), we have the revenue function for this case
 \begin{equation}\label{eqn:revenue_3}
 \overline{re}(\bar{q}_0) = \frac{\bar{q}_1}{1-\bar{q}_1} + \frac{1}{\bar{q}_0+ \bar{q}_1}-1, \quad  \bar{q}_0 \in \left[\bar{q}^{min}_0, \bar{q}^{max}_0\right].
 \end{equation}
We check the second-order conditions of a quasi-convex function, \ie, at any point with zero slope, the second derivative is non-negative:
$
re'(\bar{q}_0) = 0 \Rightarrow re''(\bar{q}_0) > 0.
$
The first derivative of the revenue function is
$$
\overline{re}'(\bar{q}_0) = \frac{\bar{q}'_1}{(1-\bar{q}_1)^2} - \frac{1+\bar{q}'_1}{ \left( \bar{q}_0 + \bar{q}_1  \right)^2},
$$
and the corresponding second derivative is
\begin{equation}\label{eqn:second_derivative1}
\overline{re}''(\bar{q}_0)   =   \frac{\bar{q}''_1}{(1-\bar{q}_1)^2} + \frac{2\times (\bar{q}'_1)^2}{(1-\bar{q}_1)^3}
   - \frac{\bar{q}''_1}{(\bar{q}_0+ \bar{q}_1 )^2} + \frac{2\times (1+\bar{q}'_1 )^2 }{(\bar{q}_0 + \bar{q}_1 )^3}.
\end{equation}
With the expression of $\bar{q}'_i  = \frac{1}{\bar{q}_0 \times \left( \frac{1}{\bar{q}_i}  + \frac{1}{(1-\bar{q}_i )^2}   \right)}$, we can calculate that
$
\bar{q}''_i = \frac{\bar{q}'_i}{\bar{q}_0 } \times  \left( -1 +g(\bar{q}_i) \right),
$
where
\begin{equation}\label{fun_g}
g(\bar{q}_i) \triangleq   \frac{1}{ \left(\frac{1}{\bar{q}_i }   + \frac{1}{(1- \bar{q}_ i)^2 } \right)^2 } \times \left( \frac{1}{\bar{q}^2_i} -\frac{2}{(1-\bar{q}_i)^3}   \right).
\end{equation}
The function $g(\bar{q})$ has one critical property needed for the later analysis: $g(\bar{q})$  is a decreasing function over the range $[0,0.5]$,
$g(0)= 1$ and  $g(0.5)= -1/3$.
We can verify this property by showing the first derivative $g'(\bar{q})$ is negative in the range  $[0,0.5]$. 
We also plot the graph of $g(\bar{q})$ in Figure~\ref{fig:functiong}.
 We then rewrite the second derivative in (\ref{eqn:second_derivative1}) as
\begin{eqnarray}
\nonumber
\overline{re}''(\bar{q}_0)  & = & \left( \frac{\bar{q}'_1 }{ (1-\bar{q}_1)^2 }  -  \frac{\bar{q}'_1}{(\bar{q}_0 + \bar{q}_1 )^2}  \right) \times \frac{1}{\bar{q}_0} \times (-1+g(\bar{q}_1 ))\\
& & + \frac{2 \times (\bar{q}'_1)^2 }{(1-\bar{q}_1)^3}  + \frac{2\times (1+\bar{q}'_1 )^2 }{(\bar{q}_0+\bar{q}_1)^3}. \label{eqn:second_derivative}
\end{eqnarray}
Since $\overline{re}'(\bar{q}_0) =  0$, we have
$$
 \frac{ \bar{q}'_1 }{ (1-\bar{q}_1)^2 }         =   \frac{1 +  \bar{q}'_1}{ ( \bar{q}_0+\bar{q}_1)^2}.
$$
Combine with the fact that $-1+g(\bar{q}_1)\leq 0$, we can further relax the second derivative in (\ref{eqn:second_derivative}):
{
\begin{eqnarray}
(\ref{eqn:second_derivative})
\nonumber
& \geq  & \frac{1+ \bar{q}'_1 }{(\bar{q}_0+\bar{q}_1)^2 } \times   \frac{1}{\bar{q}_0}   \times ( -1+ g(\bar{q}_1 ) )\\
\nonumber
& & +  \frac{1}{(\bar{q}_0 + \bar{q}_1)^2}  \times  \frac{1}{\bar{q}_0}  \times \left( \bar{q}'_1 \times (1-g(\bar{q}_1 ))   \right) \\
\nonumber
 & & + \frac{2 \times (\bar{q}'_1)^2 }{(1-\bar{q}_1)^3}  + \frac{2 \times (1+\bar{q}'_1 )^2}{(\bar{q}_0+\bar{q}_1)^3} \\
  \nonumber
   & = & \frac{1+ \bar{q}'_1 }{(\bar{q}_0+\bar{q}_1)^3}  \times \Bigg[ \left( 1 + \frac{ \bar{q}_1}{\bar{q}_0}  \right)  \times ( -1+ g(\bar{q}_1 )) \\
   \nonumber
& &     + \frac{\bar{q}_0+\bar{q}_1}{(1+ \bar{q}'_1 ) \times \bar{q}_0 }  \left( \bar{q}'_1 \times (1-g(\bar{q}_1 ) ) \right) \\
 & & +  \frac{2\times (\bar{q}_0+\bar{q}_1)^3}{(1+ \bar{q}'_1 ) \times \bar{q}_0} \times  \left(  \frac{ (\bar{q}'_1)^2\times \bar{q}_0 }{(1-\bar{q}_1)^3}  \right) + 2 \times (1+\bar{q}'_1) \Bigg]. \label{eqn:relax0}
 \end{eqnarray}
}
We next show the following three inequalities for later analysis

$\bullet$ $\frac{\bar{q}_0+\bar{q}_1}{(1+ \bar{q}'_1 ) \times \bar{q}_0 } \geq 1$,

$\bullet$ $ (-1+g(\bar{q}_1)) +2\geq 0$,

$\bullet$ $ (\bar{q}_0+\bar{q}_1)\geq 0.5$.

The first two inequalities are easy to verify. For the last inequality, we first have $(\bar{q}_0+\bar{q}_1) \geq \bar{q}_1 $. From the definition of $\bar{q}^{min}_0$, the equality $(\bar{q}_0+\bar{q}_1) = 1- \bar{q}_1$ holds when $\bar{q}_0  = \bar{q}^{min}_0$. For any $\bar{q}_0  \in \left[ \bar{q}^{min}_0,\bar{q}^{max}_0\right]$, we further have $(\bar{q}_0+\bar{q}_1)  \geq 1- \bar{q}_1$ because $\bar{q}_1 =V(\bar{q}_0 \times exp(\theta_1-1))$ is an increasing function with respective to $\bar{q}_0$. Combining these two inequalities, we have $(\bar{q}_0+\bar{q}_1) \geq \max\{ \bar{q}_1, 1-\bar{q}_1 \}$, resulting in that  $(\bar{q}_0+\bar{q}_1) \geq 0.5$.
With these three inequalities, we can further relax (\ref{eqn:relax0}):
\begin{eqnarray}
\nonumber
(\ref{eqn:relax0})  & \geq &  \frac{1+ \bar{q}'_1 }{(\bar{q}_0+\bar{q}_1)^3}  \times \Bigg[   \frac{ \bar{q}_1}{\bar{q}_0}   \times ( -1+ g(\bar{q}_1 )) +   \bar{q}'_1 \times (1-g(\bar{q}_1 ))   \\
\nonumber
 & & +  \frac{1}{2}  \times   \frac{ \bar{q}'_1}{(1-\bar{q}_1)^3}  \times \frac{1}{\frac{1}{\bar{q}_1} + \frac{1}{(1-\bar{q}_1)^2}}  + 2 \times \bar{q}'_1  \Bigg]\label{eqn:relax1} \\
  \nonumber
 & = &  \frac{1+ \bar{q}'_1 }{(\bar{q}_0+\bar{q}_1)^3}  \times  h(\bar{q}_1),
\end{eqnarray}
where we define function $h(\bar{q}_1)$ to be
$$
h(\bar{q}_1)  \triangleq  \frac{ \bar{q}_1}{\bar{q}_0}   ( -1+ g(\bar{q}_1 )) +  \bar{q}'_1   \times (1-g(\bar{q}_1 )) + \frac{1}{2} \frac{ \bar{q}'_1}{(1-\bar{q}_1)^3}  \frac{1}{\frac{1}{\bar{q}_1} + \frac{1}{(1-\bar{q}_1)^2}} +2  \bar{q}'_1.
$$
We further simplify this function as
\begin{eqnarray*}
h(\bar{q}_1)  & = & \frac{\bar{q}'_1}{  (1-\bar{q}_1)^2 }    \times \Bigg[ \left( (1-\bar{q}_1)^2 +\bar{q}_1 \right) \times (-1+g(\bar{q}_1))   \\
& &   + (3-g(\bar{q}_1)) \times  (1-\bar{q}_1)^2   + \frac{1}{2} \times \frac{1}{\frac{1}{\bar{q}_1} + \frac{1}{(1-\bar{q}_1)} -1 } \Bigg].
\end{eqnarray*}


To prove $re''(\bar{q}_0) > 0$, we only need to show that  ${h}(\bar{q}_i)> 0$.
Since  $\bar{q}_1< 0.5$, we have $g(\bar{q}_1) > -1/3$ from Figure~\ref{fig:functiong}, under which we can check  $h(\bar{q}_1) > 0$ for all  $0\leq \bar{q}_1< 0.5$ by plotting the figure of $h(\bar{q}_1)$.
 Thus, we have proved the quasi-convexity of the revenue function $\overline{re}(\bar{q}_0)$ when $\bar{q}_1< 0.5$.
\end{proof}

\section{Proof for Lemma~\ref{theo:optimal_S=2}}\label{app:sec:6}
 \begin{proof}
The basic idea behind the proof can be illustrated via the case of $k^* = 2$.   From Lemma~\ref{theo:s1}, we know the first product would always be involved.  Thus, the remaining part is to prove it is optimal to select the second product.  The revenue when the platform displays the first product and another product is  
 \begin{equation}\label{eqn:revenue_2}
 \overline{re}(\bar{q}_0) = \frac{\bar{q}_1}{1-\bar{q}_1} + \frac{1}{\bar{q}_0+ \bar{q}_1}-1, \quad  \bar{q}_0 \in \left[\bar{q}^{min}_0, \bar{q}^{max}_0\right].
 \end{equation}
We recall  that  $\overline{re}(\bar{q}^{min}_0)$ denotes the revenue of selecting the first two products,
and  $\overline{re}(\bar{q}^{max}_0)$ represents the revenue of only selecting the first product.
From~Lemma~\ref{lemma_quasi_convex_1}, we know that the revenue function $\overline{re}(\bar{q}_0)$ is  quasi-convex over the interval $\left[\bar{q}^{min}_0, \bar{q}^{max}_0 \right]$, which implies
$$
\overline{re}(\bar{q}_0) \leq \max\left\{ \overline{re}(\bar{q}^{min}_0), \overline{re}(\bar{q}^{max}_0)\right\}, \quad  \forall \bar{q}_0 \in \left[\bar{q}^{min}_0, \bar{q}^{max}_0 \right].
$$
$\overline{re}(\bar{q}^{max}_0)$ cannot be the maximum value of $\overline{re}(\bar{q}_0)$, because otherwise the size of the optimal product set is $1$, which contradicts the assumption in this lemma.  We then have
$
\overline{re}(\bar{q}_0) \leq  \overline{re}(\bar{q}^{min}_0),
$
meaning that the platform always selects the first two products when $k^*=2$. The analysis for the case of $k^*>2$ follows the same principle. Thus, we have completed  the proof to this lemma.
 \end{proof}

\section{Results for $\bar{q}_1\geq 0.5$ in Bertrand Competition Game}\label{app:sec:7}

When there is one seller that dominates the market, \ie, $\bar{q}_1 \geq 0.5$, in the revenue maximization, we need more complicated analysis to establish the optimality of the  top $k^*$ search segmentation mechanism.  We have proved the result for the case of $k^*=1$ in Lemma~\ref{theo:s1}. We now further consider the special cases of $k^*=2$ and $k^*\geq 3$ without any restriction of $q_1$, respectively.

\begin{theorem}\label{theorem:appendix1}
For revenue maximization, the optimal search segmentation mechanism is to display the first two products if the cardinality of the optimal product set is two, \ie, $k^*=2$.
\end{theorem}
\begin{proof}
The basic idea is the same as that in the proof of Lemma~\ref{theo:optimal_S=2}, showing the quasi-convexity of the revenue function $\overline{re}(\bar{q}_0)$. Without the assumption of $\bar{q}_1< 0.5$, we need some new techniques to establish this result.
\begin{lemma}
The revenue function $\overline{re}(\bar{q}_0)$ is quasi-convex over the range $[\bar{q}_{0}^{min}, \bar{q}_{0}^{max}]$.
\end{lemma}
\begin{proof}
We first calculate the derivative of $\overline{re}(\bar{q}_0)$ in (\ref{eqn:revenue_3}),
$$
\overline{re}'(\bar{q}_0) = \frac{\bar{q}'_1}{(1-\bar{q}_1)^2} -\frac{1+\bar{q}'_1}{(\bar{q}_0+ \bar{q}_1)^2}.
$$
Using the equation that $\bar{q}'_1=  \frac{1}{\bar{q}_0\times \left( \frac{1}{\bar{q}_1}+ \frac{1}{(1-\bar{q}_1)^2} \right) }$, we further have
\begin{equation}\label{eqn:revenue_derivative}
 \overline{re}'(\bar{q}_0) = A\times \left(( \bar{q}_0+2) \times \bar{q}^2_1  + ( \bar{q}_0^2 + \bar{q}_0 - 1 ) \times \bar{q}_1-  \bar{q}_0  \right),
\end{equation}
where $A = \frac{1}{((1-\bar{q}_1)^2 + \bar{q}_1 )}\times \frac{1}{\bar{q}_0} \times \frac{1}{(\bar{q}_1+\bar{q}_0)^2}$ and is positive.
We evaluate the values of  $\overline{re}'(\bar{q}_0)$  at two end points $\bar{q}^{min}_0$ and $\bar{q}^{max}_0$, respectively.
If $\bar{q}_0$ is equal to $\bar{q}^{min}_0$, representing the platform displays  the first two products, then the demands satisfy the relation:  $2\times \bar{q}_1 = 1- \bar{q}^{min}_0.$
With this equation, we can derive
$$\overline{re}'(\bar{q}^{min}_0) = A \times \left( -3\times  \bar{q}^{min}_0 - (\bar{q}^{min}_0)^3  \right) < 0.$$
Similarly, when $\bar{q}_0$ is equal to $\bar{q}^{max}_0$, we have
$ \bar{q}_1 = 1- \bar{q}^{max}_0$ and
$$\overline{re}'(\bar{q}^{max}_0) = A\times \left( 1-2\times \bar{q}^{max}_0 \right).$$
Based on the definition of $V(x)$ in (\ref{eqn:v_func}), we have the following relation
\begin{eqnarray}
\nonumber
& &  (1-\bar{q}^{max}_0) \times  exp\left( \frac{1-\bar{q}^{max}_0}{\bar{q}^{max}_0} \right) =\bar{q}^{max}_0  \times exp(\theta_1-1) \\
 \Rightarrow & & \frac{1-\bar{q}^{max}_0}{\bar{q}^{max}_0}  =  W(exp(\theta_1-1)). \label{eqn:revenue_3}
\end{eqnarray}
For such scenario, we further consider two different cases:

$\bullet$ Case A: for the first case with $\theta_1 \leq 2$, according to the equation in (\ref{eqn:revenue_3}),
we  have  $\bar{q}^{max}_0 \geq 1/2$, and thus $re'(\bar{q}^{max}_0) \leq 0$.

\begin{figure}[!tbp]
  \centering
   \includegraphics[scale=0.5]{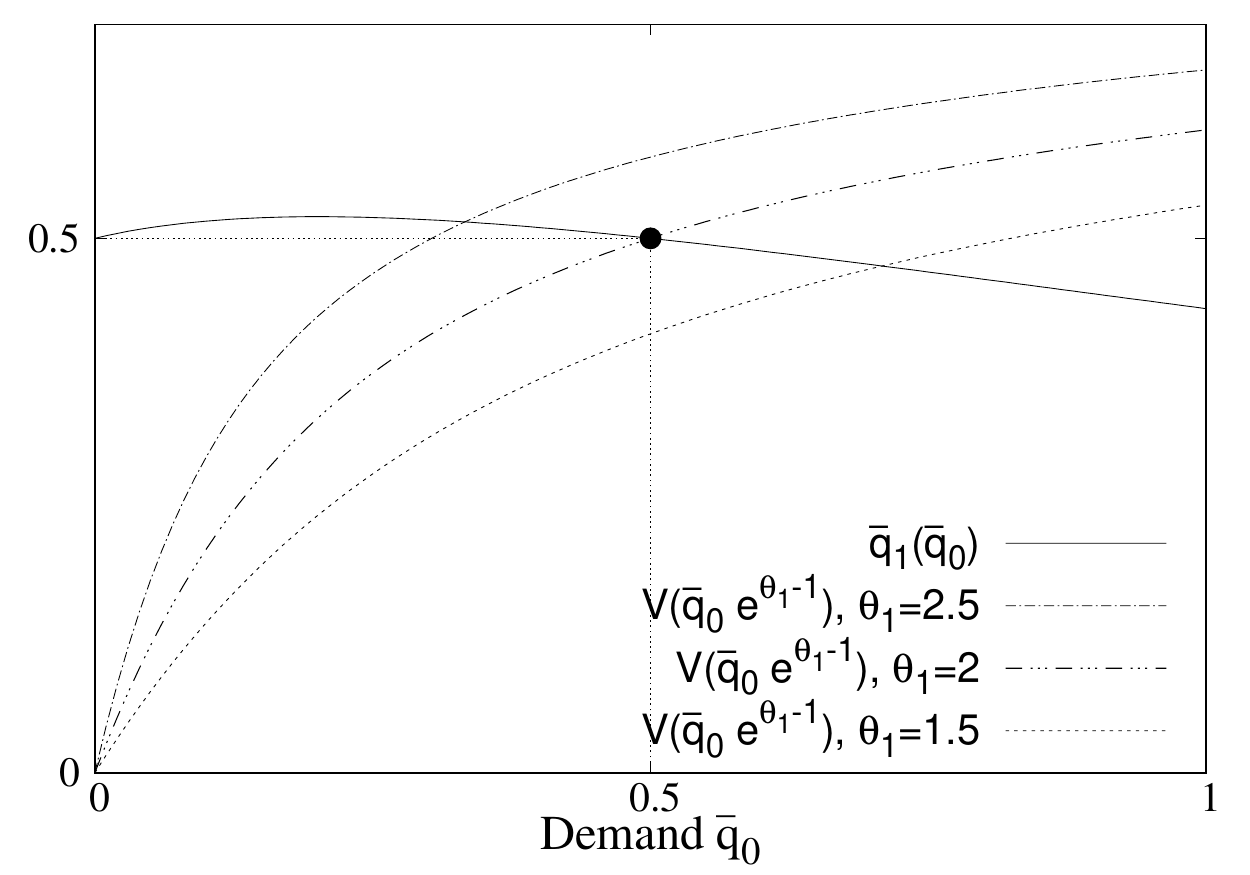}
  \caption{The  possible intersections of functions $\bar{q}_1(\bar{q}_0)$ and $V(\bar{q}_0\times \exp(\theta_1-1) )$ with different quality parameter $\theta_1$. }
  \label{fig:func_intersection}
\end{figure}

$\bullet$ Case B: for the other case with $\theta_1 > 2$, we have $\bar{q}^{max}_0 < 1/2$, and thus $re'(\bar{q}^{max}_0) > 0$.

We next show the number of roots for the equation $re'(\bar{q}_0) = 0$ in the above two cases, respectively.
The main trick here is to regard $\bar{q}_0$ as a fixed value, and express $\bar{q}_1$ as a function of $\bar{q}_0$ by setting $re'(\bar{q}_0) = 0$ in~(\ref{eqn:revenue_derivative}). Based on this idea, we have
$$
\bar{q}_1(\bar{q}_0) = \frac{-\left( \bar{q}_0^2 + \bar{q}_0 -1 \right) + \sqrt{\bar{q}_0^4  +2\times\bar{q}_0^3 +3 \times \bar{q}_0^2 +6\times \bar{q}_0+1 } }{2 \times (\bar{q}_0 + 2)}.
$$
We recall that $\bar{q}_1$ can also be expressed as
$V(\bar{q}_0\times \exp(\theta_1-1) ).$ The intersections of these two functions represent the roots for the equation $re'(\bar{q}_0) = 0$. We plot these two functions in Figure~\ref{fig:func_intersection}.

In Case A discussed  before, \ie, $\theta \leq 2$,  we can observe from the figure that the functions $\bar{q}_1(\bar{q}_0)$ and $V(\bar{q}_0\times \exp(\theta_1-1) )$ have at most one intersection, implying $re'(\bar{q}_0) = 0$ has at  most one root. This is because when $\theta\leq 2$, the intersection of these two functions can only occur in the range $\bar{q}_0 \in [0.5,1]$. We further have that $\bar{q}_1(\bar{q}_0)$ decreases, while $V(\bar{q}_0\times \exp(\theta_1-1) )$ increases in this interval $\bar{q}_0 \in [0.5,1]$. Together with the results that $re'(\bar{q}^{min}_0) < 0$ and $re'(\bar{q}^{max}_0) \leq 0$ in Case A, we can derive that the root of $re'(\bar{q}_0) = 0$ cannot be located in $[\bar{q}^{min}_0, \bar{q}^{max}_0]$ and then $re'(\bar{q}_0)\leq 0$ for $\bar{q}_0\in [\bar{q}^{min}_0, \bar{q}^{max}_0]$. This means that $\overline{re}(\bar{q}_0)$ is a non-increasing function and thus is quasi-convex in such case.

For Case B, \ie, $\theta > 2$, we can observe from Figure~\ref{fig:func_intersection} that $\bar{q}_1(\bar{q}_0)$ and $V(\bar{q}_0\times \exp(\theta_1-1) )$ have at most two intersections, and thus the equation $re'(\bar{q}_0) = 0$ could have one or two roots. In Case B, we have $re'(\bar{q}^{min}_0) < 0$ and $re'(\bar{q}^{max}_0) > 0$, and thus only one root $\bar{q}^*_0$ can locate $[\bar{q}^{min}_0, \bar{q}^{max}_0]$. This implies that $\overline{re}(\bar{q}_0)$ decreases in $[\bar{q}^{min}_0,\bar{q}^*_0]$ and increases in  $[\bar{q}^*_0, \bar{q}^{max}_0]$. Thus, $\overline{re}(\bar{q}_0)$ is also quasi-convex in Case B.

From the above discussion, we can conclude that revenue function $\overline{re}(\bar{q}_0)$ is quasi-convex over the interval $\left[\bar{q}^{min}_0, \bar{q}^{max}_0\right]$.
\end{proof}
Following the principle of the proof in Lemma~\ref{theo:optimal_S=2}, we can derive the result of this Theorem.
\end{proof}

The next theorem extends the previous result to the scenario when the cardinality of the optimal product set is larger than $2$.
\begin{theorem}\label{theorem:appendix2}
For revenue maximization, the platform will always display the top $k^*$ products if the optimal search segmentation mechanism is to select $k^*\geq 3$ products.
\end{theorem}
\begin{proof}
The basic idea behind the proof can be illustrated via the case of $k^* = 3$.  Without loss of generality, we can further assume that we have already selected the product set $S$ with the first two products $1$ and $2$. We only need to prove it is optimal to involve the third highest quality product.
The revenue function in this case is
$$
\overline{re}(\bar{q}_0) = \frac{\bar{q}_1}{1-\bar{q}_1} + \frac{\bar{q}_2}{1-\bar{q}_2} + \frac{1}{\bar{q}_0+ \bar{q}_1+ \bar{q}_2} -1.
$$
We continue the proof by considering the following two cases: $\bar{q}_1< 0.5$ and $\bar{q}_1 \geq 0.5$, respectively.

$\bullet$ Case A: in the first case, we require $\bar{q}_1 < 0.5$, which implies all the demands $\bar{q}_i$'s are less than $0.5$.
The basic idea is again to show the quasi-convexity of the revenue function $\overline{re}(\bar{q}_0)$ over the interval $\left[\bar{q}^{min}_0, \bar{q}^{max}_0\right]$, which is quite similar to the proof in Lemma~\ref{lemma_quasi_convex_1}.
 For completeness, we also present the detailed proof for the case of $k^*=3$ here.
\begin{lemma}\label{lem:quasi-convex_s=2:appendix}
For the case of $k^*=3$, the revenue function $\overline{re}(\bar{q}_0)$  is quasi-convex over the interval $\left[\bar{q}^{min}_0,\bar{q}^{max}_0\right]$ when $\bar{q}_1 < 0.5$.
\end{lemma}
\begin{proof}
We check the second-order conditions of a quasi-convex function, \ie, at any point with zero slope, the second derivative is non-negative:
$
re'(\bar{q}_0) = 0 \Rightarrow re''(\bar{q}_0) > 0.
$
The first derivative of the revenue function is
$$
\overline{re}'(\bar{q}_0) = \frac{\bar{q}'_1}{(1-\bar{q}_1)^2} + \frac{\bar{q}'_2}{(1-\bar{q}_2)^2} - \frac{1+\bar{q}'_1+ \bar{q}'_2}{ \left( \bar{q}_0 + \bar{q}_1 + \bar{q}_2 \right)^2},
$$
and the corresponding second derivative is
\begin{eqnarray}
\nonumber
\overline{re}''(\bar{q}_0)  & = &  \frac{\bar{q}''_1}{(1-\bar{q}_1)^2} + \frac{2\times (\bar{q}'_1)^2}{(1-\bar{q}_1)^3} + \frac{\bar{q}''_2}{(1-\bar{q}_2)^2} + \frac{2\times (\bar{q}'_2)^2}{(1-\bar{q}_2)^3}\\
& &   - \frac{\bar{q}''_1+ \bar{q}''_2}{(\bar{q}_0+ \bar{q}_1 + \bar{q}_2)^2} + \frac{2\times (1+\bar{q}'_1 + \bar{q}'_2)^2 }{(\bar{q}_0 + \bar{q}_1 + \bar{q}_2)^3}. \label{eqn:second_derivative1:appendix}
\end{eqnarray}
With the expression of $\bar{q}'_i  = \frac{1}{\bar{q}_0 \times \left( \frac{1}{\bar{q}_i}  + \frac{1}{(1-\bar{q}_i )^2}   \right)}$, we can calculate that
$
\bar{q}''_i = \frac{\bar{q}'_i}{\bar{q}_0 } \times  \left( -1 +g(\bar{q}_i) \right),
$
where
\begin{equation}\label{fun_g:appendix}
g(\bar{q}_i) \triangleq   \frac{1}{ \left(\frac{1}{\bar{q}_i }   + \frac{1}{(1- \bar{q}_ i)^2 } \right)^2 } \times \left( \frac{1}{\bar{q}^2_i} -\frac{2}{(1-\bar{q}_i)^3}   \right).
\end{equation}
 We then rewrite the second derivative in (\ref{eqn:second_derivative1:appendix}) as
\begin{eqnarray}
\nonumber
\overline{re}''(\bar{q}_0)  & = & \left( \frac{\bar{q}'_1 }{ (1-\bar{q}_1)^2 }  -  \frac{\bar{q}'_1}{(\bar{q}_0 + \bar{q}_1 + \bar{q}_2)^2}  \right) \times \frac{1}{\bar{q}_0} \times (-1+g(\bar{q}_1 ))\\
\nonumber
& & + \left( \frac{\bar{q}'_2 }{ (1-\bar{q}_2)^2 }  -  \frac{\bar{q}'_2}{(\bar{q}_0 + \bar{q}_1 + \bar{q}_2)^2}  \right) \times \frac{1}{\bar{q}_0} \times (-1+g(\bar{q}_2 ) )\\
& & + \frac{2 \times (\bar{q}'_1)^2 }{(1-\bar{q}_1)^3} +  \frac{2 \times (\bar{q}'_2)^2 }{(1-\bar{q}_2)^3} + \frac{2\times (1+\bar{q}'_1 + \bar{q}'_2)^2 }{(\bar{q}_0+\bar{q}_1+ \bar{q}_2)^3}. \label{eqn:second_derivative:appendix}
\end{eqnarray}
Since $\overline{re}'(\bar{q}_0) =  0$, we have
$$
 \frac{ \bar{q}'_1 }{ (1-\bar{q}_1)^2 }   +  \frac{\bar{q}'_2 }{ (1-\bar{q}_2)^2 }      =   \frac{1 +  \bar{q}'_1 + \bar{q}'_2}{ ( \bar{q}_0+\bar{q}_1+ \bar{q}_2)^2}.
$$
Combining with the fact that $-1+g(\bar{q}_1) \leq  -1+g(\bar{q}_2) \leq 0$, we can further relax the second derivative in (\ref{eqn:second_derivative:appendix}):
{
\small
\begin{eqnarray}
(\ref{eqn:second_derivative:appendix})
\nonumber
& \geq  & \frac{1+ \bar{q}'_1 + \bar{q}'_2}{(\bar{q}_0+\bar{q}_1+ \bar{q}_2)^2} \times   \frac{1}{\bar{q}_0}   \times ( -1+ g(\bar{q}_1 ) ) \\
\nonumber
& &  +  \frac{1}{(\bar{q}_0 + \bar{q}_1 + \bar{q}_2)^2}   \frac{1}{\bar{q}_0}  \left( \bar{q}'_1 \times (1-g(\bar{q}_1 ))  +  \bar{q}'_2 \times (1-g(\bar{q}_2 ) \right) \\
\nonumber
 & & + \frac{2 \times (\bar{q}'_1)^2 }{(1-\bar{q}_1)^3} +  \frac{2 \times (\bar{q}'_2)^2 }{(1-\bar{q}_2)^3} + \frac{2 \times (1+\bar{q}'_1 + \bar{q}'_2)^2}{(\bar{q}_0+\bar{q}_1+ \bar{q}_2)^3} \\
  \nonumber
   & = & \frac{1+ \bar{q}'_1 + \bar{q}'_2}{(\bar{q}_0+\bar{q}_1+ \bar{q}_2)^3}  \times \Bigg[ \left( 1 + \frac{ \bar{q}_1}{\bar{q}_0} + \frac{\bar{q}_2}{\bar{q}_0} \right)  \times ( -1+ g(\bar{q}_1 )) \\
 \nonumber
 & & + \frac{\bar{q}_0+\bar{q}_1+ \bar{q}_2}{(1+ \bar{q}'_1 + \bar{q}'_2) \times \bar{q}_0 }  \left( \bar{q}'_1 \times (1-g(\bar{q}_1 ))  +  \bar{q}'_2 \times (1-g(\bar{q}_2 )) \right) \\
 \nonumber
 & & +  \frac{2\times (\bar{q}_0+\bar{q}_1+ \bar{q}_2)^3}{(1+ \bar{q}'_1 + \bar{q}'_2) \times \bar{q}_0} \times  \left(  \frac{ (\bar{q}'_1)^2\times \bar{q}_0 }{(1-\bar{q}_1)^3} +  \frac{(\bar{q}'_2)^2 \times \bar{q}_0 }{(1-\bar{q}_2)^3}  \right) \\
& & + 2 \times (1+\bar{q}'_1 + \bar{q}'_2) \Bigg]. \label{eqn:relax0:appendix}
 \end{eqnarray}
\normalsize
}
We next show the following three inequalities for later analysis

$\bullet$ $\frac{\bar{q}_0+\bar{q}_1+ \bar{q}_2}{(1+ \bar{q}'_1 + \bar{q}'_2) \times \bar{q}_0 } \geq 1$,

$\bullet$ $ (-1+g(\bar{q}_1)) +2\geq 0$,

$\bullet$ $ (\bar{q}_0+\bar{q}_1+ \bar{q}_2)\geq 0.5$.

The first two inequalities are easy to verify. For the last inequality, we first have $(\bar{q}_0+\bar{q}_1+ \bar{q}_2) \geq \bar{q}_1 $. From the definition of $\bar{q}^{min}_0$, the equality $(\bar{q}_0+\bar{q}_1+ \bar{q}_2)  = 1- \bar{q}_1$ holds when $\bar{q}_0  = \bar{q}^{min}_0$. For any $\bar{q}_0  \in \left[ \bar{q}^{min}_0,\bar{q}^{max}_0\right]$, we further have $(\bar{q}_0+\bar{q}_1+ \bar{q}_2)  \geq 1- \bar{q}_1$ because $\bar{q}_i =V(\bar{q}_0 \times exp(\theta_i-1))$ is an increasing function with respective to $\bar{q}_0$. Combining these two inequalities, we have $(\bar{q}_0+\bar{q}_1+ \bar{q}_2) \geq \max\{ \bar{q}_1, 1-\bar{q}_1 \}$, resulting in that  $(\bar{q}_0+\bar{q}_1+ \bar{q}_2) \geq 0.5$.
With these three inequalities, we can further relax (\ref{eqn:relax0:appendix}):
\begin{eqnarray}
\small
\nonumber
(\ref{eqn:relax0:appendix})  & \geq &  \frac{1+ \bar{q}'_1 + \bar{q}'_2}{(\bar{q}_0+\bar{q}_1+ \bar{q}_2)^3}  \times \Bigg[  \left(  \frac{ \bar{q}_1}{\bar{q}_0} + \frac{\bar{q}_2}{\bar{q}_0} \right)  \times ( -1+ g(\bar{q}_1 )) \\
 \nonumber
& & +   \bar{q}'_1 \times (1-g(\bar{q}_1 ))  +  \bar{q}'_2 \times (1-g(\bar{q}_2 ))  \\
\nonumber
 & & +  \frac{1}{2}      \left(  \frac{ \bar{q}'_1}{(1-\bar{q}_1)^3}  \frac{1}{\frac{1}{\bar{q}_1} + \frac{1}{(1-\bar{q}_1)^2}} +  \frac{\bar{q}'_2 }{(1-\bar{q}_2)^3}  \frac{1}{\frac{1}{\bar{q}_2} + \frac{1}{(1-\bar{q}_2)^2}}  \right)  \\
 \nonumber
& & + 2 \times (\bar{q}'_1 + \bar{q}'_2) \Bigg]\label{eqn:relax1} \\
  \nonumber
 & = &  \frac{1+ \bar{q}'_1 + \bar{q}'_2}{(\bar{q}_0+\bar{q}_1+ \bar{q}_2)^3}  \times  (h(\bar{q}_1)  +  h(\bar{q}_2)),
\end{eqnarray}
\normalsize
where we define function $h(\bar{q}_i), i\in S$ to be
$$
h(\bar{q}_i)  \triangleq  \frac{ \bar{q}_i}{\bar{q}_0}   ( -1+ g(\bar{q}_1 )) +  \bar{q}'_i   \times (1-g(\bar{q}_i )) + \frac{1}{2} \frac{ \bar{q}'_i}{(1-\bar{q}_i)^3}  \frac{1}{\frac{1}{\bar{q}_i} + \frac{1}{(1-\bar{q}_i)^2}} +2  \bar{q}'_i.
$$
We further simplify this function as
\begin{eqnarray*}
h(\bar{q}_i)  & = & \frac{\bar{q}'_i}{  (1-\bar{q}_i)^2 }    \times \Bigg[ \left( (1-\bar{q}_i)^2 +\bar{q}_i \right) \times (-1+g(\bar{q}_1))   \\
& &   + (3-g(\bar{q}_i)) \times  (1-\bar{q}_i)^2   + \frac{1}{2} \times \frac{1}{\frac{1}{\bar{q}_i} + \frac{1}{(1-\bar{q}_i)} -1 } \Bigg].
\end{eqnarray*}

\begin{figure}[!tbp]
  \centering
  \includegraphics[scale=0.5]{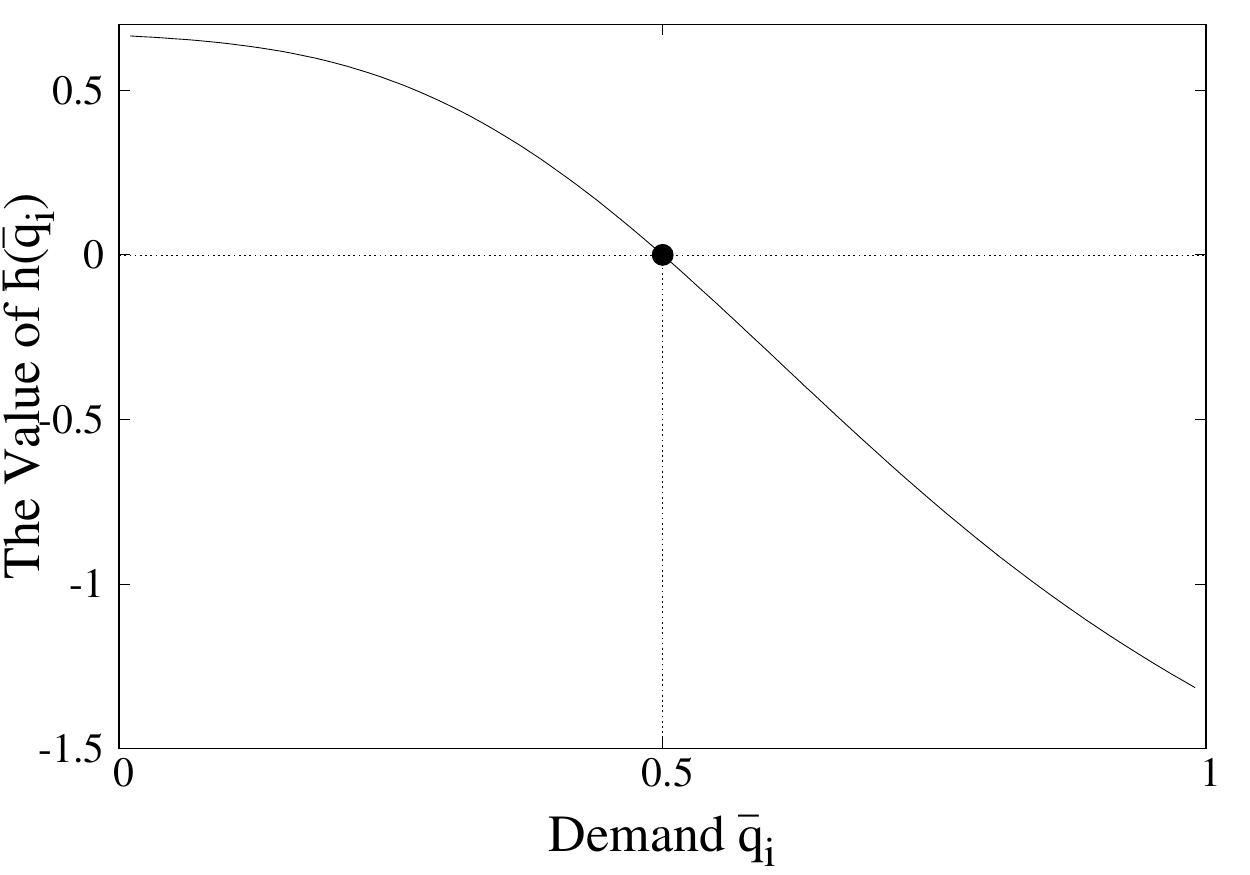}
  \caption{The graph of function $\bar{h}(\bar{q}_i)$, which shows that  $\bar{h}(\bar{q}_i)\geq 0$ for all $\bar{q}_i\leq 0.5$.}
  \label{fig:function_g}
\end{figure}


To prove $re''(\bar{q}_0) > 0$, we only need to show that  ${h}(\bar{q}_i) > 0$.
Since  $\bar{q}_1< 0.5$ in this case, we have $g(\bar{q}_1) > -1/3$, leading to that
$$
h(\bar{q}_i) > \bar{q}'_i  \times (1-\bar{q}_i)^2  \times \bar{h}(\bar{q}_i),
$$
where
\begin{equation}\label{eqn:h_bar}
\bar{h}(\bar{q}_i) \triangleq  -\frac{4}{3} \left( (1-\bar{q}_i)^2 +\bar{q}_i \right) + (3-g(\bar{q}_i)) \times  (1-\bar{q}_i)^2   + \frac{1}{2}  \frac{1}{\frac{1}{\bar{q}_i} + \frac{1}{(1-\bar{q}_i)} -1 }.
\end{equation}
It is tedious but one can verify $\bar{h}(\bar{q}_i) > 0$  for all  $0\leq \bar{q}_i< 0.5$, and a graph of this function is shown in Figure~\ref{fig:function_g}.  Thus, we have proved the quasi-convexity of revenue function $\overline{re}(\bar{q}_0)$ when $\bar{q}_1< 0.5$.
\end{proof}

$\bullet$ Case B: in the second case, $\bar{q} _1$ is larger or equal to $0.5$. We show that the optimal mechanism would still select the third product in this case.
Let $\bar{q}^*_0$ be the $\bar{q}_0$ that makes  $\bar{q}_1 = 0.5$.  According to the definition of $\bar{q}^{min}_0$, when $\bar{q}_0  = \bar{q}^{min}_0$, we have $\bar{q}_1< 0.5$, and thus we can get $\bar{q}^{min}_0 < \bar{q}^*_0$.
This threshold $\bar{q}^*_0$ separates the range $[\bar{q}^{min}_0, \bar{q}^{max}_0]$ into two intervals: $[\bar{q}^{min}_0, \bar{q}^*_0]$ and  $[\bar{q}^*_0, \bar{q}^{max}_0]$.

We next show that it is not necessary to consider the interval $[\bar{q}^*_0, \bar{q}^{max}_0]$ under the assumption of $k^* =3$.
We prove this by showing that the revenue generated when $\bar{q}_0 \in [\bar{q}^*_0, \bar{q}^{max}_0]$ is always less than the revenue of only selecting the first product, leading to a contradiction to $k^*=3$.
Since $\bar{q}_1 = V(\bar{q}_0 \times exp(\theta_1-1))$ is an increasing function with respective to $\bar{q}_0$, we have $\bar{q}_1\geq 0.5$ for  $\bar{q}_0 \in [\bar{q}^*_0, \bar{q}^{max}_0]$. The candidate  product $j\in \mathbb{S} \backslash S$ satisfies equation
$
\bar{q}_1 +\bar{q}_2 +\bar{q}_j = 1-\bar{q}_0
$
with the requirements of $\bar{q}_1\geq 0.5$ and $\bar{q}_0 \in [\bar{q}^*_0, \bar{q}^{max}_0]$. The revenue of selecting these three products is
$$
\overline{re}_{a}(\bar{q}_0) = \frac{\bar{q}_1}{1-\bar{q }_1} + \frac{\bar{q}_2}{1-\bar{q }_2} + \frac{\bar{q}_j}{1-\bar{q }_j}.
$$
For each such new product $j$, we construct a virtual product ${j}'$ with the quality parameter $\theta_{j}'$ that satisfies
$\bar{q}_{j'}   = V(\bar{q}_0\times exp(\theta_{j'}-1)) =  \bar{q}_2+\bar{q}_j  $ and $\bar{q}_1 +  \bar{q}_{j'} = 1-\bar{q}_0$.
The revenue of selecting the first product $1$ and the virtual product $j'$ is
$$
\overline{re}_b(\bar{q}_0) = \frac{\bar{q}_1}{1-\bar{q}_1} + \frac{\bar{q}_{j'}}{1-\bar{q}_{j'}} = \frac{\bar{q}_1}{1-\bar{q}_1} + \frac{1}{\bar{q}_0 +\bar{q}_1}-1 .
$$
We have shown in Lemma~\ref{lemma_quasi_convex_1} that such revenue function $re_b(\bar{q}_0)$ is quasi-convex over the range $\left[\bar{q}^{min'}_0,\bar{q}^{max'}_0\right]$\footnote{We note that $\left[\bar{q}^{min'}_0, \bar{q}^{max'}_0\right]$ is different from $\left[\bar{q}^{min}_0, \bar{q}^{max}_0\right]$. When $\bar{q}_0 = \bar{q}^{max'}_0$, we only select the first product, and in the case of
$\bar{q}_0 = \bar{q}^{max}_0$, we select the first two products.
Thus, we have $\bar{q}_1 = 1-\bar{q}^{max'}_0$ and $\bar{q}_1 +\bar{q}_2 = 1-\bar{q}^{max}_0$, resulting in that $\bar{q}^{max'}_0 \geq \bar{q}^{max}_0$. With the definition of $\bar{q}^{min'}_0$, we still have $\bar{q}^{min'}_0 < \bar{q}^*_0$.}, which implies that $re_b(\bar{q}_0)$ is also quasi-convex over $[\bar{q}^{*}_0,\bar{q}^{max'}_0]$. Thus, we have
$$
\overline{re}_b(\bar{q}_0) \leq \max \left\{ \overline{re}_b(\bar{q}^*_0), \overline{re}_b(\bar{q}^{max'}_0) \right\},  \quad \forall \bar{q}_0  \in [\bar{q}^{*}_0,\bar{q}^{max'}_0].
$$
We next show that $\overline{re}_b(\bar{q}^{max'}_0) \geq \overline{re}_b(\bar{q}^*_0)$. We first specifically write down the revenue at these two end-points, respectively.
When $\bar{q}_0=\bar{q}^{max'}_0$, meaning the platform only selects the first product,
and we have $V(\bar{q}^{max'}_0 \times exp(\theta_i-1) ) = \bar{q}_1 = 1-\bar{q}^{max'}_0$.
Similar to (\ref{eqn:revenue_3}), we can also have
$
 \frac{1-\bar{q}^{max'}_0}{\bar{q}^{max'}_0}  = W(exp(\theta_1-1)),
$
where $W(x)$ is the Lambert function. The revenue for the case of $\bar{q}_0=\bar{q}^{max'}_0$ is
$$
\overline{re}_b({\bar{q}^{max'}}_0)  = \frac{\bar{q}_1}{1-\bar{q}_1} =  \frac{1-\bar{q}^{max'}_0}{\bar{q}^{max'}_0} = W(exp(\theta_1-1)).
$$
When $\bar{q}_0  = \bar{q}^*_0$, we select the virtual product $j'$ such that $\bar{q}_1=0.5$. Thus, we have $V(\bar{q}^*_0\times exp(\theta_1-1))=0.5$, and by the definition of function $V(x)$, we can further derive
$
\bar{q}^*_0 = 0.5 \times exp(2-\theta_1).
$
The revenue for the case of $\bar{q}_0  = \bar{q}^*_0$ is
$$
\overline{re}_b(\bar{q}^*_0) = \frac{\bar{q}_1}{1-\bar{q}_1} + \frac{1}{\bar{q}^*_0 +\bar{q}_1}-1 = \frac{2}{2\times \bar{q}^*_0 +1} = \frac{2}{exp(2-\theta_1)+1}.
$$
To prove $\overline{re}_b(\bar{q}^{max}_0) \geq \overline{re}_b(\bar{q}^*_0)$, we only need to show that
$$W(exp(\theta_1-1)) \geq  \frac{2}{exp(2-\theta_1)+1}, \quad \forall \theta_1\geq 0.$$
We introduce an auxiliary notation $x \triangleq exp(\theta_1-1)$, and express the above inequality as
$$W(x) \geq \frac{2\times x}{e+x}, \quad \forall x\geq exp(-1),$$ where $e$ is Euler's number.
As $W(x)\geq 2 > \frac{2\times x}{e+x} $  for $x\geq 2 \times exp(2)$, the remaining part is just to show that
$$W(x)\geq \frac{2\times x}{e+x}, \quad \forall  x \in  [exp(-1), 2\times exp(2)].$$
The above inequality can be verified by more tedious calculations, and a plot of $W(x)$ and $\frac{2\times x}{e+x}$ over the range $x\in [exp(-1), 2\times exp(2)]$ is shown in Figure~\ref{fig:function_w}. From the above discussion, we can derive that for any $\bar{q}_0 \in [\bar{q}^*_0, \bar{q}^{max}_0]$, the following relation holds
 \begin{eqnarray*}
 \overline{re}_b(\bar{q}^{max'}_0) &  \geq  & \overline{re}_b(\bar{q}_0)  = \frac{\bar{q}_1}{1-\bar{q }_1}  +  \frac{\bar{q}_2  + \bar{q}_j}{1-(\bar{q}_2+\bar{q}_j)}  \\
&  \geq &  \frac{\bar{q}_1}{1-\bar{q }_1} + \frac{\bar{q}_2}{1-\bar{q }_2} + \frac{\bar{q}_j}{1-\bar{q }_j} = \overline{re}_{a}(\bar{q}_0).
 \end{eqnarray*}
This means that when $\bar{q}_0 \in [\bar{q}^*_0, \bar{q}^{max}_0]$, the revenue of displaying three products $re_a(\bar{q}_0)$  is always less than the revenue of only displaying the first product  $ re_b(\bar{q}^{max'}_0)$, which contradicts the assumption of $k^* =3$. Thus, we can conclude that $\bar{q}_0 \notin [\bar{q}^*_0, \bar{q}^{max}_0]$.

\begin{figure}[!tbp]
  \centering
  \includegraphics[scale=0.5]{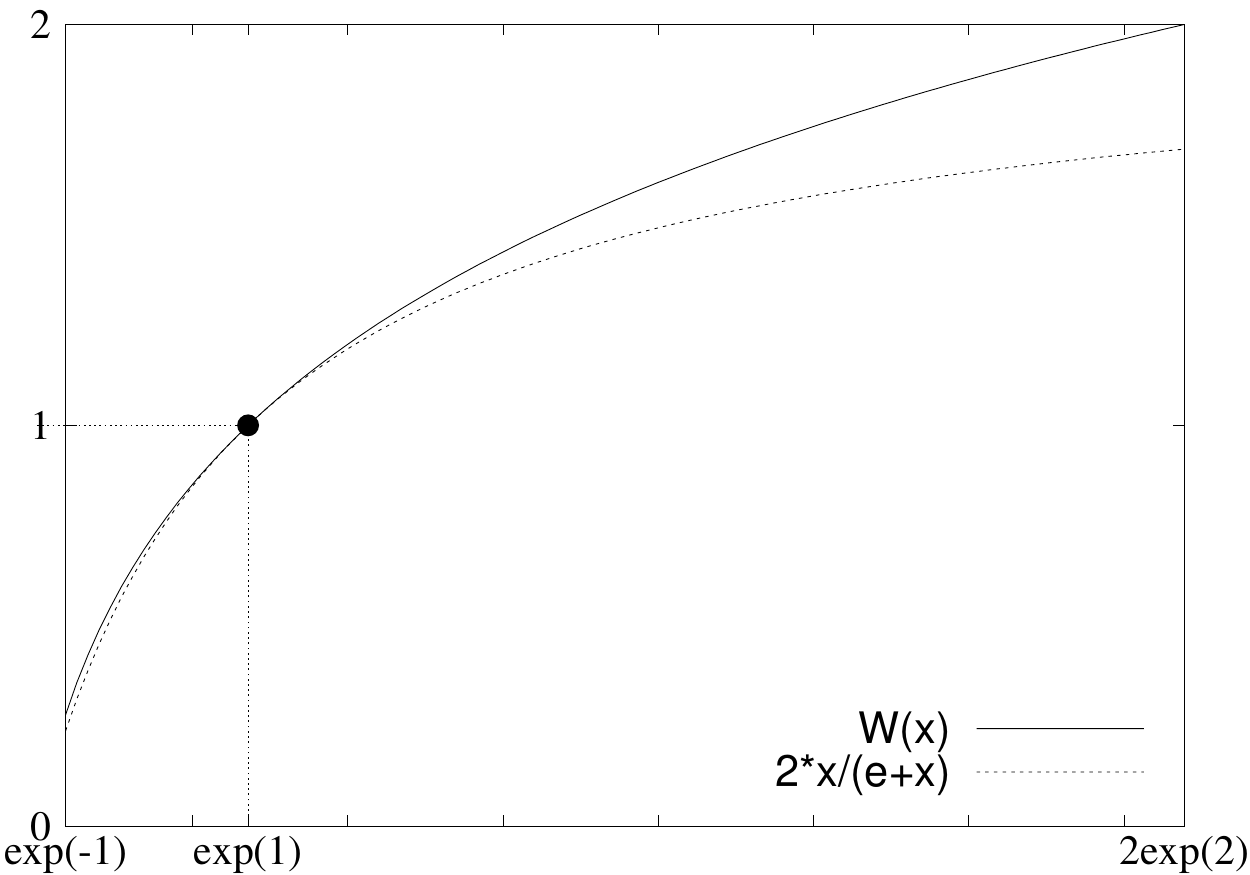}
  \caption{The graph of functions $W(x)$ and $\frac{2 \times x}{e+x}$, which shows that  $W(x)\geq \frac{2 \times x}{e+x}$ for all $exp(-1) \leq x\leq 2\times exp(2)$.}
  \label{fig:function_w}
\end{figure}

We next consider the other interval $\bar{q}_0 \in [\bar{q}^{min}_0, \bar{q}^{*}_0]$, and show that
it is still optimal to display  the third  product given the selected product set $S=\{1,2\}$.
In this case, the market share $\bar{q}_1$ is less than $0.5$. Thus, we can use the result in Lemma~\ref{lem:quasi-convex_s=2:appendix} to conclude that the revenue function $\overline{re}(\bar{q}_0)$ is quasi-convex over  $[\bar{q}^{min}_0, \bar{q}^*_0]$, and have
$$
\overline{re}(\bar{q}_0) \leq \max\left\{ \overline{re}(\bar{q}^{min}_0), \overline{re}(\bar{q}^*_0)\right\}, \quad  \forall \bar{q}_0 \in \left[\bar{q}^{min}_0, \bar{q}^{*}_0 \right].
$$
Similarly,  $\overline{re}(\bar{q}^*_0)$ cannot be the maximum value of $\overline{re}(\bar{q}_0)$ for  $\bar{q}_0 \in \left[\bar{q}^{min}_0, \bar{q}^{*}_0 \right]$, because otherwise the cardinality of the optimal product set would be $1$ due to the previous discussion.
We then have $\overline{re}(\bar{q}_0) \leq \overline{re}(\bar{q}^{min}_0)$ for $\bar{q}_0 \in \left[\bar{q}^{min}_0, \bar{q}^{*}_0 \right]$, meaning that it is  optimal to select the third product in this case.

From the above discussion for Case A and Case B, we have proved that the optimal search segmentation mechanism is to select the first three products when $k^* = 3$.
\end{proof}

From Theorem~\ref{theorem:appendix1} and Theorem~\ref{theorem:appendix2}, we have shown that the top-$k$ search segmentation mechanism is still optimal when $q_i$'s can be arbitrage values.

\section{Proof for Lemma~\ref{lem:social_cournot}}\label{app:sec:8}
\begin{proof}
In the case of $k^*=1$, we can verify from (\ref{nsw_cournot}) that the equilibrium social welfare increases with the quality of the selected product. Thus, the optimal mechanism is to display the first product when $k=1$. We assume $k\geq 2$ for the following discussion.
We recall that products are sorted in a non-decreasing order in terms of quality. As the Lambert function is an increasing function over $[0,+\infty)$, we further have $w_1\geq w_2 \geq \cdots \geq w_n >0$, where we recall $w_i  = W(exp(\theta_i-1))$.
Suppose the optimal search segmentation mechanism selects the set of products $S\subseteq \mathbb{S}$ with $ |S| = k$, which does not contain all the first $k$ products.
We will show that we can replace this product set
 $S$ with the set
$\hat{S}=\{1,2,\cdots, k\}$, and also improve equilibrium social welfare.
Consider product $j \in \mathbb{S}\backslash S$, that achieves the maximum
$w_t$'s among the unselected products, \ie, ${w}_j = \argmax \{w_t | t\in \mathbb{S}\backslash S \}$.
As the selected set $S$ does not contains all the first $k$ products, there must exist one selected product $i\in S$ such that $w_i < {w}_j$\footnote{For the case that $w_i = w_j$ and $j<i$, we can directly exchange the product $i$ and product $j$, and obtain  the same equilibrium social welfare. }.

We now show that we can replace product $j$ with product $i$ to improve social welfare. Motivated by the equilibrium social welfare $\widehat{sw}$ in (\ref{nsw_cournot}), we introduce a function
\begin{equation}\label{func:g_w}
g(w) \triangleq \log(1+\sum_{t\in S \backslash \{i\}} w_t + w ) + \frac{\sum_{t\in S\backslash \{i\}}( w^2_t + w_t) + (w^2+w)}{1+\sum_{t\in S\backslash \{i\}} w_t +w}.
\end{equation}
We note that the equilibrium social welfare of selecting product set $S$ is $g(w_i)$, the equilibrium social welfare of selecting product set $S\backslash\{i\}\cup\{j\}$ (\ie, replacing product $i$ with product $j$) is $g(w_j)$, and $g(0)$ denotes the social welfare of selecting $k-1$ products $S\backslash\{i\}$.

The key idea to prove this lemma is to show that $g(w)$ is quasi-convex over the interval $[0, {w}_j]$, which implies that
$$
g(w) \leq \max \left\{g(0) , g(w_j)\right\}, \quad \forall w\in [0,w_j].
$$
Assuming the quasi-convexity of $g(w)$, we claim $g(0)$ cannot be the maximum value of $g(w)$ under the assumption that the optimal search segmentation mechanism is to involve $k$ products. Suppose $g(0) \geq g(w)$ for all $w\in [0,w_j]$. This means that the set $S\backslash \{i\}$ with cardinality $k-1$ achieves higher social welfare than the set $S$ with cardinality $k$, which contradicts the assumption in this lemma. Thus, the maximum $g(w)$ over the interval $[0,w_j]$ is $g(w_j)$ and $g(w_j)\geq g(w_i)$, meaning that selecting product $j$ instead of product $i$ achieves higher social welfare.

We now prove the quasi-convexity of the function  $g(w)$.
We first calculate the derivative of  $g(w)$
$$
g'(w)= \frac{w^2 + B \times w + C}{(1+\sum_{t\in S \backslash \{i\}} w_t +w)^2},
$$
where $B = (2 \sum_{t\in S \backslash \{i\}} w_t+3 )$ and $C= -\sum_{t\in S \backslash \{i\}} \left( w^2_t-w_t\right) +2$.  $B$ is always positive, while $C$ could be positive, negative or zero. We continue the proof by considering the following two cases.

$\bullet$ If $C\geq0$ then $g'(w)$ is positive for any non-negative $w$, meaning $g(w)$ is an increasing function, and then is a quasi-convex function over the range $[0,w_j]$. 

$\bullet$ If $C<0$ then the equation $g'(w)=0$ has a positive root $w^*=(-B+\sqrt{B^2-4C})/2$.
We claim that  $g'(w)$ cannot be negative for all $w\in [0,w_j]$, otherwise $g(0)$ would be the maximum value for all $g(w)$, which contradicts the assumption of the lemma due to the same reasons discussed before.
Thus, we can have that $g'(w)$ is negative over the range $[0,w^*)$ and positive in $[w^*,w_j]$, \ie, $g(w)$ decreases in $[0,w^*)$ and increases in $[w^*,w_j]$. With this property, it follows that $g(w)$ is also quasi-convex in this case.

From the above discussion, we have proved that if the optimal mechanism is to display $k$ products, $g(w_j)$ is always larger than $g(w_i)$.
This means that for any selected product set $S$ that does not contain the first $k$ products,
we can always find an unselected  product $j\in \mathbb{S}\backslash S$ and a selected product $i\in S$ with $w_j>w_i$, and improve social welfare by replacing product $j$ with product $i$. Iteratively conducting this operation, we can obtain a new set $\hat{S}$ that exactly contains the  first $k$ products, and achieve the maximum equilibrium social welfare when the number of the optimal products is $k$.
\end{proof}

\section{Proof for Lemma~\ref{lem_revenue_cournot}}\label{app:sec:9}
\begin{proof}
The basic idea is also to check the quasi-convexity of the equilibrium revenue function. The only difference is to change $g(w)$ in (\ref{func:g_w}) to $\hat{g}(w)$ based on the equilibrium revenue in (\ref{nre_cournot}),  \ie,
\begin{equation}
\hat{g}(w) \triangleq  \frac{\sum_{t\in S\backslash \{i\}}( w^2_t + w_t) + (w^2+w)}{1+\sum_{t\in S\backslash \{i\}} w_t +w}.
\end{equation}
The corresponding derivative of $\hat{g}(w)$ is
$$
\hat{g}'(w) = \frac{w^2 + \hat{B} \times w + \hat{C}}{(1+\sum_{t\in S \backslash \{i\}} w_t +w)^2},
$$
where $\hat{B}  =  2 \times \left(\sum_{t\in S \backslash \{i\}} w_t +1 \right)$ and $\hat{C} = 1-\sum_{t\in S \backslash \{i\}}  w^2_t$. Here $\hat{B}$ is always positive, while $\hat{C}$ could be positive, negative or zero. The following steps are similar to those in the proof for Lemma~\ref{lem:social_cournot}, and we omit them here.
\end{proof}

\end{document}